\documentclass[envcountsame,envcountsect,a4paper,runningheads]{llncs}

\newcommand\arxiv[1]{#1}
\newcommand\conf[1]{}

\usepackage{hyperref}

\hypersetup{
  colorlinks=true,
  linkcolor=blue,
  anchorcolor=blue,
  citecolor=blue,
  urlcolor=blue,
  filecolor=blue,
  breaklinks=true
}
\usepackage{prftree}
\usepackage{amsmath,amsfonts,amssymb}
\usepackage[only,llbracket,rrbracket]{stmaryrd}
\SetSymbolFont{stmry}{bold}{U}{cmr}{b}{n} 
\usepackage{cite}
\usepackage{float}
\floatstyle{boxed}
\restylefloat{figure}
\usepackage{multicol}
\allowdisplaybreaks
\usepackage{nicefrac}

\newcommand{\lvec}{\ensuremath{\textrm{Vectorial}}}
\newcommand{\llin}{\ensuremath{\textrm{Lineal}}}
\newcommand{\lvecr}{\ensuremath{\lvec_{\textrm{\tiny R}}}}
\newcommand{\llinr}{\ensuremath{\llin_{\textrm{\tiny R}}}} 
\newcommand{\ve}[1]{\ensuremath{\mathbf{#1}}}
\newcommand{\Sc}{\ensuremath{\mathsf{S}}}
\DeclareMathAlphabet{\mathpzc}{OT1}{pzc}{m}{n}
\newcommand{\varu}[1]{\ensuremath{\mathpzc{#1}}}
\newcommand{\vara}[1]{\ensuremath{\mathbb{#1}}}
\newcommand{\sui}[1]{\ensuremath{\sum_{i=1}^{#1}}}
\newcommand{\suj}[1]{\ensuremath{\sum_{j=1}^{#1}}}
\newcommand{\suk}[1]{\sum_{k=1}^{#1}} \newcommand{\FV}[1]{\ensuremath{FV}(#1)}
\newcommand{\true}{{\textbf{true}}} \newcommand{\false}{{\textbf{false}}}
\newcommand{\True}{\mathcal{T}} \newcommand{\False}{\mathcal{F}}
\newcommand{\canon}[1]{\left[#1\right]}
\newcommand{\cocanon}[1]{\left\{#1\right\}}
\newcommand{\citeCharTypes}{\cite[Lem.~4.2]{vectorial}}
\newcommand{\citeEquivSumScalars}{\cite[Lem.~4.4]{vectorial}}
\newcommand{\citeEquivForall}{\cite[Lem.~4.5]{vectorial}}
\newcommand{\V}{\mathcal{V}} \newcommand{\ssubt}{\mathrel{\preceq}}
\newcommand{\tnorm}[1]{\ensuremath{\mathrel{\mathcal{W}\left(#1\right)}}}
\newcommand{\recap}[3]{\noindent {\bf #1 \ref{#2}.} \emph{#3}}
\newcommand{\xrecap}[4]{\noindent {\bf #1 \ref{#3} (#2).} \emph{#4}}
\newcommand{\haligneddots}{\cleaders\hbox to 3pt{\hss$\cdot$\hss}\hfill\kern0pt}
\newcommand{\textbetweenrules}[2][2pt]{%
  \par\addvspace{\topsep}
  \noindent\makebox[\textwidth]{%
    \sbox0{\quad#2\quad}%
    \dimen0=.5\dimexpr\ht0+#1\relax \dimen2=-.5\dimexpr\ht0-#1\relax
    \dimen4=.5\dimexpr\textwidth-\wd0\relax \setbox2=\vbox to \ht0{%
      \vss \hbox to \dimen4{\haligneddots} \vss }%
    \copy2 \box0 \box2 }\par\nopagebreak\addvspace{\topsep}%
}
\newcommand{\inductioncase}[1]{\noindent\ignorespaces\textbetweenrules{\textbf{#1}}}
\newcommand{\sug}[2]{\ensuremath{\sum_{#1=1}^{#2}}}
\newcommand{\textleadbydots}[2][2pt]{%
  \par\addvspace{\topsep}
  \noindent\makebox[\textwidth]{%
    \sbox0{#2\quad}%
    \dimen0=.5\dimexpr\ht0+#1\relax \dimen2=-.5\dimexpr\ht0-#1\relax
    \dimen4=\dimexpr\textwidth-\wd0\relax \setbox2=\vbox to \ht0{%
      \vss \hbox to \dimen4{\haligneddots} \vss }%
    \box0 \box2 }\par\nopagebreak\addvspace{\topsep}%
}
\newcommand{\subst}[2]{[{#1}/{#2}]}

\begin{document}

\title{The Vectorial Lambda Calculus Revisited\thanks{This paper is based on the first author's master thesis~\cite{Noriega20}.}}
\author{Francisco Noriega\inst{1} \and
Alejandro D\'iaz-Caro\inst{2,3}}

\authorrunning{F.~Noriega and A.~D\'iaz-Caro}

\institute{DC, FCEyN. Universidad de Buenos Aires, Argentina\\
  \email{frannoriega.92@gmail.com}
  \and
  DCyT. Universidad Nacional de Quilmes, Argentina
  \and
  ICC.  CONICET--Universidad de Buenos Aires, Argentina
  \\
  \email{adiazcaro@icc.fcen.uba.ar}}

\maketitle              

\begin{abstract}
  We revisit the \lvec\ Lambda Calculus, a typed version of \llin. \lvec\ (as
  well as \llin) has been originally designed for quantum computing, as an extension to
  System F where linear combinations of lambda terms are also terms and linear
  combinations of types are also types.
  In its first presentation, \lvec\ only provides a weakened version of the
  Subject Reduction property. We prove that our revised \lvec\ Lambda Calculus
  supports the standard version of said property, answering a long standing
  issue.  In addition we also introduce the concept of weight of types and
  terms, and prove a relation between the weight of terms and of its types.

\end{abstract}

\section{Introduction}\label{ch:introduction}
The ``quantum data, classic control'' paradigm has been proposed for
programming languages by Selinger~\cite{selinger_2004}. It presumes that quantum computers will have a specialized
device, known as QRAM~\cite{Knill04}, attached to a classical computer, with the
latter instructing the former which operations to perform over which qubits. In
this scheme, the classical computer is the one that reads the output of measurements
performed on the qubits to retrieve the classical bits and continue running the
program. Hence, the quantum memory and the allowed operations are only provided
as black boxes under this paradigm. The quantum lambda
calculus~\cite{selinger_valiron_2006}, as well as several programming languages for
describing quantum algorithms, such as Qiskit~\cite{Qiskit}, or the more evolved Quipper~\cite{quipper} and QWIRE~\cite{qwire},
follow this scheme. However, a lambda calculus allowing for programming those
black boxes continues to be a long-standing problem. This is what is known as
``quantum data and control''. One of the first attempts for quantum control
within the lambda calculus was van Tonder's calculus~\cite{tonder04lambda},
which placed the lambda terms directly inside the quantum memory. A completely
different path started with Arrighi and Dowek's
work~\cite{ArrighiDowekRTA08,ArrighiDowekLMCS17}, who proposed a new untyped
calculus called \llin. In \llin, linear combinations (i.e.~superpositions) of 
terms are also terms, and they showed how to encode quantum operations with it.

\llin\ is a minimalistic language able to model high-level computation with linear
algebra, providing a computational definition of vector spaces and
bilinear functions. The first problem addressed by this language was how to
model higher-order computable operators over infinite dimensional vector spaces.
This serves as a basis for studying wider notions of computability upon
abstract vector spaces, whatever the interpretation of the vectors might be
(probabilities, number of computational paths leading to one result, quantum
states, etc.). Thus, the terms are modelled as said state vectors, and if
$\ve{t}$ and $\ve{u}$ are valid terms, then so is the term \( \alpha\cdot\ve{t}
+ \beta\cdot\ve{u} \), representing the superposition of the state vectors
$\ve{t}$ and $\ve{u}$ with some scalars $\alpha$ and $\beta$. However, the
downside of this generality in the context of quantum computing, is that the
operators are not restricted to being unitary (as required by quantum physics).
It was not until several years later~\cite{DiazcaroGuillermoMiquelValironLICS19} that the problem of
how to restrict such a language to the quantum realm has been somehow solved using a realizability technique.
However, such a technique is based on defining the denotational semantics
first, and then extracting a type system from there (maybe with an infinite
amount of typing rules) that fits such semantics. The problem on how to
extract a finite set of typing rules, which is expressive enough, remains open.
In~\cite{MalherbeDiazcaro2020b} there is a first attempt to define such a
language, called Lambda-$\mathcal S_1$, which is, however, far from the
original \llin. For example, in \llin\ it is possible to define an Oracle
$U_f$ implementing the one-bit to one-bit function $f$\footnote{See, for
  example, \cite[\S 1.4.2]{NC10} for more information about this Oracle, or
  \cite[\S 5.2]{NC10}, for a deeper discussion about oracles in general.} with
  a lambda-term abstraction taking the function $f$ as a parameter. This is not
  possible in Lambda-$\mathcal S_1$, since to ensure that the produced $U_f$ is
  unitary would require to do a test of orthogonality between two open
  lambda-terms. Therefore, the realizability technique provides only part of
  the solution, but more complex type systems that take into account the scalars
  within the types might be needed to solve this problem.

\lvec~\cite{vectorial} is a polymorphic
typed version of \llin\ providing a formal account of linear operators and
vectors at the level of the type system, including both scalars and sums of
types. In \lvec, if $\Gamma \vdash \ve t:T$ and $\Gamma \vdash \ve r:R$ then
$\Gamma \vdash~\alpha\cdot\ve t + \beta\cdot\ve r~:~\alpha\cdot T + \beta\cdot
R$. In general, if $\ve t$ has type $\sum_i\alpha_i\cdot U_i$, it reduces to a superposition $\sum_i\alpha_i\cdot\ve r_i$, with each $\ve r_i$ of type $U_i$.
As in \llin, finite vectors and
matrices can be encoded within \lvec. The linear
combinations of types typing the encoded expressions give some information on
the linear combination of values to be obtained. In particular, $U_f$ is
typable in \lvec. In addition, \llin, its untyped version, required some kind of
restrictions to avoid non confluent terms issued from the fact that not
normalising terms can be considered as a form of infinite, and so the
subtraction of any two terms is not always well defined\footnote{An easy
  example is a term $Y_{\ve b}$ rewriting to $\ve b+Y_{\ve b}$, so without
  further restrictions, $Y_{\ve b}-Y_{\ve b}$ may be rewritten both to $\ve 0$
  and to $\ve b+Y_{\ve b}-Y_{\ve b}$ and thus to $\ve b$.}. With type systems
  ensuring strong normalisation, such kind of issues
  disappear~\cite{ArrighiDiazcaroLMCS12,DiazcaroPetitWoLLIC12,lmcs:927,vectorial}.

\lvec\ has been a step into the quest for a quantum lambda calculus in the
quantum data and control paradigm. However, despite its many interesting
properties, \lvec\ does not feature the subject reduction property. For
example, while $(\lambda x.x)+(\lambda x.x)$ can be typed by $(U\to U)+(V\to
V)$ for any $U$ and $V$, $2\cdot(\lambda x.x)$ can only by typed by
$2\cdot(U\to U)$ or $2\cdot (V\to V)$. So, even if $U+U$ is equivalent to
$2\cdot U$, subject reduction is lost if $\ve t+\ve t$ reduces to $2\cdot\ve
t$, as it is the case in \llin.
In~\cite{vectorial} only a weakened
version of subject reduction has been established. This is the reason why,
after defining \lvec, the quest for quantum control in the lambda calculus has
taken a turn into simpler type
systems~\cite{DiazcaroDowekTPNC17,DiazcaroMalherbeLSFA18,DiazcaroGuillermoMiquelValironLICS19,DiazcaroDowekRinaldiBIO19,DiazcaroMalherbeACS20,DiazcaroDowek2020b},
none of them considered to be complete yet.

By revisiting \lvec, we noticed that it is possible to fix its lack of
subject reduction, while preserving many properties of the original system.
This is the main contribution of our paper: to provide a non-trivial
redefinition of \lvec, featuring subject reduction, while still having the main
desirable properties of the original system. We think that this modified version of
\lvec\ will provide the needed framework in the quest for the
quantum-controlled lambda calculus.

\subsection*{Plan of the paper}\label{sec:introduction:plan}
The definition of this revised version of \lvec, which we will call \lvecr\ along
this paper to avoid confusion, is given in Section~\ref{ch:vecrev}. We also discuss the design decisions
behind the revision in order to regain the standard version of the subject
reduction property. In Section~\ref{sec:examples} we bring back key examples
from \lvec, showing that they are still valid for \lvecr. We prove subject
reduction in Section~\ref{ch:sr}. In Section~\ref{ch:other-properties} we
present the proof for other desirable properties of the system: progress, strong
normalisation, and weight preservation, that is, the weight of a typed term is equal to the weight of its type.

\section{The calculus}\label{ch:vecrev}
\subsection{\texorpdfstring{\llinr}{LinealR}: The untyped setting}\label{sec:vecrev:terms}
\llin~\cite{ArrighiDowekRTA08,ArrighiDowekLMCS17} extends the lambda calculus
with linear combinations of terms. In our revised version, which we call \llinr,
the grammar of terms is given by
\[
  \ve t ::= x\mid\lambda x.\ve t\mid (\ve t)~\ve t\mid\alpha\cdot\ve t\mid\ve
  t+\ve t
\]
where $\alpha$ belongs to a commutative ring $(\Sc,+,\times)$.

This grammar differs from that of \llin\ in the fact that we do not include a
term $\ve 0$ representing the null linear combination. Indeed, $0\cdot\ve t$ is
a proper term, but it differs from $0\cdot\ve r$ when $\ve t\neq\ve r$. This
modification comes from the fact that in a typed calculus, $\ve 0$ would have to be
typed with any type. Then, for example, $(\lambda x.x+0\cdot\ve t)~\ve r$ may
not have a type, if $\ve t$ is not an arrow type for example, while $(\lambda
x.x)~\ve r$ can always be typed. So it becomes crucial not to simplify the term
$0\cdot\ve t$, and consequently we do not need a term $\ve 0$. In fact, such
linear combinations can be seen as forming a ``weak'' module, differing from a
module in the fact that there is no neutral element for the addition.
See~\cite[\S II.B]{DiazcaroGuillermoMiquelValironLICS19} for a longer
discussion about the weak structure, which has been historically used within
the concept of unbounded operators, introduced by von Neumann to give a
rigorous mathematical definition to the operators that are used in quantum
mechanics. For historical reasons we will continue calling the calculus ``The
\emph{Vectorial} Lambda Calculus'', while it could be named ``The \emph{Weak
Module} Lambda Calculus''.

Variables and abstractions are called basis
terms~\cite{ArrighiDowekRTA08,vectorial} or pure values~\cite{DiazcaroGuillermoMiquelValironLICS19}:
\[
  \ve b ::= x\mid\lambda x.\ve t
\]

The reduction rules, given in Figure~\ref{fig:vecrev:terms}, are split in four groups.
The groups E (elementary rules) and F (factorisation rules) deal with the (weak)
module axioms. The group B is composed by only one rule, the beta-reduction,
following a ``call-by-basis'' strategy~\cite{AssafDiazcaroPerdrixTassonValironLMCS14}, that is, the beta-reduction can occur
only when the argument is a basis term. Finally, the group A (application rules)
deals with applications in linear combinations: If the left hand side or the
right hand side of an application is a linear combination (and so, the
conditions for applying the call-by-basis beta-rule are not met), then the
application is first distributed over the linear combination.

\begin{figure}[t]
  \centering
  \[
    \begin{array}{rcl@{\quad}rcl@{\quad}rcl}
      \multicolumn{3}{c}{\textrm{Group E}}&\multicolumn{3}{c}{\textrm{Group F}}&\multicolumn{3}{c}{\textrm{Group A}}\\
      1\cdot\ve{t}&\to&\ve{t} &            \alpha\cdot\ve{t}+\beta\cdot\ve{t}&\to&(\alpha+\beta)\cdot\ve{t} & (\ve{t}+\ve{r})~\ve{u}&\to&(\ve{t})~\ve{u}+(\ve{r})~\ve{u}\\
      \alpha\cdot (\beta\cdot\ve{t})&\to&(\alpha\times\beta)\cdot\ve{t} &           \alpha\cdot\ve{t}+\ve{t}&\to&(\alpha+1)\cdot\ve{t} & (\ve{t})~(\ve{r}+\ve{u})&\to&(\ve{t})~\ve{r}+(\ve{t})~\ve{u}\\
      \alpha\cdot (\ve{t}+\ve{r})&\to&\alpha\cdot\ve{t}+\alpha\cdot\ve{r} &           \ve{t}+\ve{t}&\to&(1+1)\cdot\ve{t} & (\alpha\cdot\ve{t})~\ve{r}&\to&\alpha\cdot (\ve{t})~\ve{r}\\
      &&&&&&                        (\ve{t})~(\alpha\cdot\ve{r})&\to&\alpha\cdot(\ve{t})~\ve{r}\\
      \multicolumn{9}{c}{\textrm{Group B}}\\
      &&& (\lambda x.\ve{t})~\ve{b}&\to&\ve{t}[\ve{b}/x] 
    \end{array}
  \]

  Contextual rules
  \[
    \prftree{\ve t\to\ve r}{\alpha\cdot\ve t\to\alpha\cdot\ve r}
    \quad
    \prftree{\ve t\to\ve r}{\ve u+\ve t\to \ve u+\ve r}
    \quad
    \prftree{\ve t\to \ve r}{(\ve u)~\ve t\to(\ve u)~\ve r}
    \quad
    \prftree{\ve t\to \ve r}{(\ve t)~\ve u\to(\ve r)~\ve u}
    \quad
    \prftree{\ve t\to \ve r}{\lambda x.\ve t\to\lambda x.\ve r}
  \]
  \caption{Reduction relation of \llinr\ and \lvecr.}
  \label{fig:vecrev:terms}
\end{figure}

\subsection{\texorpdfstring{\lvecr}{VectorialR}: Typed \texorpdfstring{\llinr}{LinealR}}\label{sec:vecrev:typesystem}
The grammar of types~\cite{vectorial} consists in a
sort of unit types, that is, types which are not linear combinations of other
types, aimed to type base terms, and a sort of general types, which are linear
combinations of unit types, or type variables of that sort.
\[
  \begin{array}[t]{l@{\hspace{1.5cm}}r@{$\ ::=\quad$}l}
    \text{\em Types:} & T & U~|~\alpha\cdot T\mid T+T\mid \vara{X}\\
    \text{\em Unit types:} & U & \varu{X}\mid U\to T\mid\forall\varu{X}.U\mid\forall \vara{X}.U
  \end{array}
\]
We write $T,R,S$ for general types and $U, V, W$ for unit types. Notice that
there are two kinds of variables, distinguished by its typography. Variables
$\varu X, \varu Y, \varu Z$ are variables meant to be replaced only by unit
types, while $\vara X, \vara Y, \vara Z$ can be replaced by any type. Note, however, that,
for example, $\forall\varu X.\varu X$ is a valid type (even if not inhabited),
while $\forall\vara X.\vara X$ is not even grammatically correct. In the same
way, since arrows have the shape $U\to T$, an $\vara X$ variable can only appear
in the body of the arrow. The shape of the arrow accounts for the fact that the
calculus is call-by-base, and so only base terms can be passed as arguments.

As with terms, types form a (weak) module. Therefore, we consider the
equivalence between types given in Figure~\ref{fig:typeequiv}.
\begin{figure}[t]
    \[
      \begin{array}[t]{r@{~\equiv~}l@{\hspace{1.5cm}}r@{~\equiv~}l}
        1\cdot T & T							 & 	\alpha\cdot T+\beta\cdot T & (\alpha+\beta)\cdot T\\
        \alpha\cdot(\beta\cdot T) & (\alpha\times\beta)\cdot T & 		T+R & R+T\\
        \alpha\cdot T+\alpha\cdot R	&\alpha\cdot (T+R) 			 & 	T+(R+S) & (T+R)+S
      \end{array}
    \]
    \caption{Equivalence between types}
    \label{fig:typeequiv}
\end{figure}

A typing sequent $\Gamma\vdash\ve t:T$ relates a context $\Gamma$, formed by a
set of unit-typed term variables (and, as usual, written as a coma-separated list
of variables and types), a term $\ve t$ and a type $T$. The rules to construct
valid typing sequents are given in Figure~\ref{fig:vecrev:types}, and they have been
modified in relation to the set of rules from \lvec~\cite{vectorial}. We write
$X$ when we do not want to specify which kind of variable we refer to
($\varu{X}$ or $\vara{X}$). The notation $T[A/X]$ is a way to abbreviate two
rules, one where $A$ is a unit type and $X$ is $\varu X$, and another one with $A$
any type and $X$ is $\vara X$. Similarly, $\forall_I$ (resp.~$\forall_E$) stands for
$\forall_\varu{I}$ or $\forall_\vara{I}$ (resp.~$\forall_\varu{E}$ or
$\forall_\vara{E}$) depending on which kind of variable is being introduced
(resp.~eliminated).

\begin{figure}[t]
  \centering
  \[
    \begin{array}{c}
      \vcenter{\prftree[r]{$ax$}{}{\Gamma, x:{U}\vdash x:{U}}}
      \qquad
      \qquad
      \vcenter{\prftree[r]{$\equiv$}{\Gamma\vdash\ve{t}: T}{R \equiv T}{\Gamma\vdash\ve{t}: R}}
      \\
      \vcenter{\prftree[r]{$\to_I$}{\Gamma, x:{U} \vdash\ve{t}: T}{\Gamma \vdash \lambda x.\ve{t}:{U}\to T}}
      \quad
      \vcenter{\prftree[r]{$\to_E$}{\Gamma \vdash\ve{t}:\sui{n}\alpha_i\cdot\forall\vec{X}.(U\to T_i)}{\Gamma\vdash\ve{r}:\suj{m}\beta_j\cdot U[\vec{A}_j/\vec{X}]}{\Gamma \vdash(\ve{t})~\ve{r}:\sui{n}\suj{m} \alpha_i\times\beta_j\cdot {T_i[\vec{A}_j/\vec{X}]}}}
      \\
      \vcenter{\prftree[r]{$\forall_{I}$}{\Gamma\vdash\ve{t}: \sui{n}\alpha_i\cdot U_i}{X\notin\FV{\Gamma}}{\Gamma\vdash\ve{t}:\sui{n}\alpha_i\cdot \forall X.U_i}}
      \qquad
      \qquad
      \vcenter{\prftree[r]{$\forall_{E}$}{\Gamma\vdash\ve{t}: \sui{n}\alpha_i\cdot \forall X.U_i}{\Gamma\vdash\ve{t}: \sui{n}\alpha_i\cdot U_i[A/X]}}
      \\
      \vcenter{\prftree[r]{$+_I$}{\Gamma\vdash\ve{t}: T}{\Gamma\vdash\ve{r}: R}{\Gamma\vdash\ve{t}+\ve{r}: T+R}}
      \qquad
      \vcenter{\prftree[r]{$1_E$}{\Gamma \vdash 1\cdot \ve{t}: T}{\Gamma \vdash \ve{t}: T}}
      \qquad
      \vcenter{\prftree[r]{$S$}{\Gamma \vdash \ve{t}: T_i}{\forall i \in \{1,\dots,n\}}{\Gamma \vdash \left(\sui{n} \alpha_i\right) \cdot \ve{t}: \sui{n} \alpha_i \cdot T_i}}
    \end{array}
  \]
  \caption{Typing rules of \lvecr.}
  \label{fig:vecrev:types}
\end{figure}

Since the main focus of this work is to provide a revision of $\lvec$ to recover
the subject reduction property, we deemed necessary to revise the typing rules.
To make it clear how this new type system solves the problem, we analyse the
problem the original system had.

In \lvec, instead of $1_E$ and $S$, there is an arguably more natural rule
$\alpha_I$:
\[
  \prftree[r]{$\alpha_I$}{\Gamma\vdash\ve t:T}{\Gamma\vdash\alpha\cdot\ve
    t:\alpha\cdot T}
\]

However, consider a term $\ve{t}$ typable both by $T$ and $R\not\equiv T$. The
term $\alpha \cdot \ve{t} + \beta \cdot \ve{t}$ can be typed by $\alpha \cdot T
+ \alpha \cdot R$, both, in \lvec\ and in \lvecr. However, upon reducing this
term by rule $\alpha \cdot \ve{t} + \beta \cdot \ve{t} \to (\alpha + \beta)
\cdot \ve{t}$ (from Group F), the given term in \lvec\ can only be typed either
by $(\alpha+\beta)\cdot T$ or $(\alpha+\beta)\cdot R$, breaking subject
reduction. Instead, the added rule $S$ in \lvecr\ allows to type such a term
with the correct type $\alpha\cdot T+\beta\cdot R$.

We can generalise the problem, so for any term $\ve{t}$ that can be typed with
$T_1,\dots,T_n$, the system should be able to type $(\sui{n} \alpha_i)
\cdot \ve{t}$ with $\sui{n} \alpha_i \cdot T_i$. The only condition
we must satisfy is that the scalar associated with the term is equal to the sum
of the scalars of the type, which in this case is $\sui{n} \alpha_i$.

Rule $S$ has been introduced to solve this problem, and it also served as a
replacement for rule $\alpha_I$, which is the particular case with $n=1$.

However, the rule $S$ alone is not enough to solve the problem. Continuing with
the example, using the new rule $S$ we have
\[
  \prftree[r]{$S$}
  {\prftree
    {\vdots}
    {\Gamma \vdash \ve{t}: T}}
  {\prftree
    {\vdots}
    {\Gamma \vdash \ve{t}: R}}
  {\Gamma \vdash (\alpha + \beta) \cdot \ve{t}: \alpha \cdot T + \beta \cdot R}
\]
In the particular case when $\alpha + \beta = 1$, the previous conclusion is
$\Gamma \vdash 1 \cdot \ve{t}: \alpha \cdot T + \beta \cdot R$, and so by
applying the rewriting rule $1 \cdot \ve{t} \to \ve{t}$ (from Group E), we end
up having to derive $\Gamma \vdash \ve{t}: \alpha \cdot T + \beta \cdot R$. Such is the reason
for the rule $1_E$.

\section{Interpretation of typing judgements}\label{sec:examples}
In the general case the calculus can represent infinite-dimensional linear
operators such as $\lambda x.x$, $\lambda x.\lambda y.y$, $\lambda x.\lambda
f.(f)\,x$,\dots and their applications. Even for such general terms $\ve t$, the
vectorial type system provides much information about the superposition of basis
terms $\sum_i\alpha_i\cdot\ve b_i$ to which $\ve t$ is reduced to, as proven
by~Theorem~\ref{thm:progress} (Progress). How much information is brought by the type
system in the finitary case is the topic of this section.

Next we show how to encode finite-dimensional 
linear operators, i.e.~matrices, together with their applications to vectors.
This encoding slightly differs from that of $\lvec$~\cite[\S 6]{vectorial}.

\subsection{In 2 dimensions}
In this section we show how $\lvecr$ handles the Hadamard gate\footnote{The
  Hadamard gate is a well known quantum operator sending $|0\rangle$ to $\frac 1{\sqrt 2}|0\rangle+\frac 1{\sqrt 2}|1\rangle$ 
  and $|1\rangle$ to $\frac 1{\sqrt 2}|0\rangle-\frac 1{\sqrt 2}|1\rangle$.},
and how to encode matrices and vectors in general.

With an empty typing context, the booleans $\true=\lambda x.\lambda y.x\,$ and
$\,\false=\lambda x.\lambda y.y$ (or $|0\rangle$ and $|1\rangle$ in Dirac notation) can be respectively typed with the types
$\True=\forall \varu{XY}.\varu X\to (\varu Y\to\varu X)\,$ and
$\,\False=\forall\varu{XY}.\varu X\to (\varu Y\to\varu Y)$. The superposition
has the following type $\vdash\alpha\cdot\true+\beta\cdot\false:\alpha\cdot\True
+ \beta\cdot\False$. (Note that it can also be typed with $(\alpha+\beta)\cdot
\forall\varu X.\varu X\to\varu X\to\varu X$).

The linear map $\ve{U}$ sending $\true$ to $a\cdot\true+b\cdot\false$ and
$\false$ to $c\cdot\true+d\cdot\false$ 
is written as
\[
  \ve U={\lambda
    x.\cocanon{((x)~\canon{a\cdot\true+b\cdot\false})~\canon{c\cdot\true+d\cdot\false}}}.
\]
where $\canon{\ve t}$ stands for $\lambda x.\ve t$, for a fresh variable $x$,
and $\cocanon{\ve t}$ stands for $(\ve t)~\lambda x.x$. This way,
$\cocanon{\canon{\ve t}}\to^*\ve t$.
Such an encoding is needed to freeze the distribution of an application with
respect to its argument. Indeed, $(\ve t)~(\ve r+\ve s)\to(\ve t)~\ve r+(\ve
t)~\ve s$, while $(\ve t)~(\lambda x.\ve s+\ve t)$ does not distribute since
the argument is already a base term.

The following sequent is valid:
\[
  \vdash\ve{U}:\forall \vara{X}.((I\to (a\cdot\True+b\cdot\False))\to(I\to
  (c\cdot\True+d\cdot\False))\to I\to \vara{X})\to \vara{X}.
\]
or, using a similar notation $\canon T$ for $I\to T$,
\[
  \vdash\ve{U}:\forall
  \vara{X}.(\canon{a\cdot\True+b\cdot\False}\to\canon{c\cdot\True+d\cdot\False}\to\canon{\vara
    X})\to \vara{X}.
\]
One can check that $\vdash~({\textbf{U}})~\true~:~a\cdot\True+b\cdot\False$, as
expected since it reduces to $a\cdot\true+b\cdot\false$:
\begin{align*}
  &({\textbf{U}})~\true\\
  &= \left(\lambda x.\cocanon{\left((x)\canon{a\cdot\true+b\cdot\false}\right)\canon{c\cdot\true+d\cdot\false}}\right)~\left(\lambda x.\lambda y.x\right)\\
  &= \lambda x.\left(\left(\left((x)~\left(\lambda f.a\cdot\true+b\cdot\false\right)\right)~\left(\lambda g.c\cdot\true+d\cdot\false\right)\right)~\left(\lambda x.x\right)\right)~(\lambda x.\lambda y.x)\\
  &\to (((\lambda x.\lambda y.x)~(\lambda f.a\cdot\true+b\cdot\false))~(\lambda g.c\cdot\true+d\cdot\false))~(\lambda x.x)\\
  &\to ((\lambda y.\lambda f.a\cdot\true+b\cdot\false)~(\lambda g.c\cdot\true+d\cdot\false))~(\lambda x.x)\\
  &\to (\lambda f.a\cdot\true+b\cdot\false)~(\lambda x.x)\\
  &\to a\cdot\true+b\cdot\false
\end{align*}

The Hadamard gate $\textbf{H}$ is the particular case $a=b=c=-d=\nicefrac1{\sqrt2}$.
The term $({\textbf{H}})~(\nicefrac1{\sqrt2}\cdot\true+\nicefrac1{\sqrt2}\cdot \false)$ has 
type $\True+0\cdot\False$, and reduces as follows.
\begin{align*}
  &({\textbf{H}})~\left(\nicefrac1{\sqrt2}\cdot\true+\nicefrac1{\sqrt2}\cdot \false\right)
  \ \to^*\ \left(({\textbf{H}})~\left(\nicefrac1{\sqrt2}\cdot \true\right)\right)+\left(({\textbf{H}})~\left(\nicefrac1{\sqrt2}\cdot \false\right)\right)\\
  &\to^*\  \nicefrac1{\sqrt2}\cdot(({\textbf{H}})~\true)+\nicefrac1{\sqrt2}\cdot(({\textbf{H}})~\false)\\
  &\to^*\  \nicefrac1{\sqrt2}\cdot\left(\nicefrac1{\sqrt2}\cdot\true+\nicefrac1{\sqrt2}\cdot\false\right)+\nicefrac1{\sqrt2}\cdot\left(\nicefrac1{\sqrt2}\cdot\true-\nicefrac1{\sqrt2}\cdot\false\right)\\
  &\to^*\  \nicefrac{1}{2}\cdot\true+\nicefrac{1}{2}\cdot\false + \nicefrac{1}{2}\cdot\true-\nicefrac{1}{2}\cdot\false
  \ \to^*\  \true + 0\cdot\false
\end{align*}

But we can do more than typing $2$-dimensional vectors or $2\times2$-matrices:
using the same technique we can encode vectors and matrices of any size.

\subsection{Vectors in \texorpdfstring{$n$}{n} dimensions}\label{sec:vec}
The $2$-dimensional space is represented by the span of $\lambda x_1x_2.x_1$ and
$\lambda x_1x_2.x_2$: the $n$-dimensional space is simply represented by the
span of all the $\lambda x_1\cdots{}x_n.x_i$, for $i \in
\left\{1,\dots,n\right\}$. As for the two dimensional case where
\[
  \vdash~
  \alpha_1\cdot\lambda x_1x_2.x_1 +
  \alpha_2\cdot\lambda x_1x_2.x_2
  ~:~
  \alpha_1\cdot\forall \varu{X}_1\varu{X}_2.\varu{X}_1
  +
  \alpha_2\cdot\forall \varu{X}_1\varu{X}_2.\varu{X}_2,
\]
an $n$-dimensional vector is typed with
\[
  \vdash~
  \sui{n}\alpha_i\cdot\lambda x_1\cdots{}x_n.x_i
  ~:~
  \sui{n}\alpha_i\cdot\forall \varu{X}_1\cdots{}\varu{X}_n.\varu{X}_i.
\]
We use the notations
\[
  {\ve e}_i^n = \lambda x_1\cdots{}x_n.x_i,
  \qquad
  {\ve E}_i^n = \forall \varu{X}_1\cdots{}\varu{X}_n.\varu{X}_i
\]
and write
\[
  \begin{array}{r@{~=~}l@{~=~}l}
    \left\llbracket
    \begin{array}{c}
      \alpha_1 \\
      \vdots \\
      \alpha_n
    \end{array}
    \right\rrbracket^{\textrm{term}}_{n}
    &
    \left(\begin{array}{c}
      \alpha_{1}\cdot{\ve e}_1^n\\
      +\\
      \cdots\\
      +\\
      \alpha_{n}\cdot{\ve e}_n^n
    \end{array}\right)
    &
    \sum\limits_{i=1}^{n}\alpha_i\cdot {\ve e}_i^n
    \\[4em]
    \left\llbracket
    \begin{array}{c}
      \alpha_1 \\
      \vdots \\
      \alpha_n
    \end{array}
    \right\rrbracket^{\textrm{type}}_{n}
    &
    \left(\begin{array}{c}
      \alpha_{1}\cdot{\ve E}_1^n \\
      +\\
      \cdots \\
      +\\
      \alpha_{n}\cdot{\ve E}_n^n
    \end{array}\right)
    &
    \sum\limits_{i=1}^{n}\alpha_i\cdot {\ve E}_i^n
  \end{array}
\]

\subsection{\texorpdfstring{$n\times m$}{nxm} matrices}\label{sec:mat}
Once the representation of vectors is chosen, it is easy to generalise the
representation of $2\times 2$ matrices to the $n\times m$ case. Suppose that the
matrix $U$ is of the form
\[
  U =
  \left(
  \begin{array}{ccc}
    \alpha_{11} & \cdots & \alpha_{1m}
    \\
    \vdots && \vdots
    \\
    \alpha_{n1} & \cdots & \alpha_{nm}
  \end{array}
  \right),
\]
then its representation is
\[
  \left\llbracket
  U
  \right\rrbracket^{\textrm{term}}_{n\times m}
  ={~~~~}
  \lambda x.
  \left\{
  \left(
  \cdots
  \left(
  (x)
  \left[
  \begin{array}{c}
    \alpha_{11}\cdot{\ve e}_1^n
    \\+\\
    \cdots
    \\+\\
    \alpha_{n1}\cdot{\ve e}_n^n
  \end{array}
  \right]
  \right)
  \cdots
  \left[
  \begin{array}{c}
    \alpha_{1m}\cdot{\ve e}_1^n
    \\+\\
    \cdots
    \\+\\
    \alpha_{nm}\cdot{\ve e}_n^n
  \end{array}
  \right]
  \right)
  \right\}\qquad
\]
and its type is
\[
  \left\llbracket
  U
  \right\rrbracket^{\textrm{type}}_{n\times m}
  ={~~~~}
  \forall\vara{X}.
  \left(
  \left[
  \begin{array}{c}
    \alpha_{11}\cdot{\ve E}_1^n
    \\+\\
    \cdots
    \\+\\
    \alpha_{n1}\cdot{\ve E}_n^n
  \end{array}
  \right]\to
  \cdots
  \to
  \left[
  \begin{array}{c}
    \alpha_{1m}\cdot{\ve E}_1^n
    \\+\\
    \cdots
    \\+\\
    \alpha_{nm}\cdot{\ve E}_n^n
  \end{array}
  \right]\to
  [~\vara{X}~]
  \right)
  \to
  \vara{X},
\]
that is, an almost direct encoding of the matrix $U$.

\section{Subject Reduction}\label{ch:sr}
Recovering the Subject Reduction property constitutes
the main focus of this work. In the original system, the Group F was the group
of rules that required special consideration and did not satisfy the property in
full.

The proof of the Subject Reduction theorem requires some intermediate results
that we develop in this section. 
\conf{We give enough details for reproducing all the proofs. The full detailed proofs are given in the 51-pages long arXiv'ed version at~\cite{paper:arxiv}.}
\arxiv{The full proofs are given in the Appendix~\ref{app:proofsSR}.}

\arxiv{We use the standard notation for equivalence classes: $[x]$ identifies the class
from which $x$ is a representative.} Given a type derivation tree $\pi$, we may
refer to it simply by its last sequent, $\pi = \Gamma \vdash \ve{t}: T$, when
there is no ambiguity. We also write $size(\pi)$ for the number of sequents
present on the tree $\pi$.

The following lemma gives a canonical form for types.
\begin{lemma}[Characterisation of types~\citeCharTypes]\label{lem:sr:typecharact}
  For any type $T$, there exist $n,m\in\mathbb{N}$, $\alpha_1,\dots,\alpha_n$,
  $\beta_1,\dots,\beta_m\in\Sc$, distinct unit types $U_1,\dots,U_n$ and
  distinct general variables $\vara{X}_1,\dots,\vara{X}_m$ such that \(
  T\equiv\sui{n}\alpha_i\cdot U_i+\suj{m}\beta_j\cdot\vara{X}_j \).
  \conf{\qed}
\end{lemma}
\arxiv{\begin{proof}
  Structural induction on $T$.
  The full details are given in Appendix~\ref{app:proofsSR}.
\end{proof}}

\arxiv{Our system admits weakening and strengthening, as stated by the following lemma.
\begin{lemma}[Weakening and Strengthening]\label{lem:sr:weakening}
  Let $\ve t$ be such that $x\not\in\FV{\ve t}$. Then $\Gamma\vdash\ve t:T$ is
  derivable if and only if $\Gamma,x:U\vdash\ve t:T$ is derivable.
\end{lemma}
\begin{proof}
  By a straightforward induction on the type derivation.
\end{proof}}

{The following two lemmas present some properties of the equivalence relation.
\begin{lemma}[Equivalence between sums of distinct elements (up to
  $\equiv$)~\citeEquivSumScalars]\label{lem:sr:equivdistinctscalars}
  Let $U_1,\dots,U_n$ be a set of distinct (not equivalent) unit types, and let
  $V_1,\dots,V_m$ be also a set distinct unit types. If $\sui{n}\alpha_i\cdot
  U_i\equiv\suj{m}\beta_j\cdot V_j$, then $m=n$ and there exists a permutation
  $p$ of $m$ such that $\forall i$, $\alpha_i=\beta_{p(i)}$ and $U_i\equiv
  V_{p(i)}$.
  \conf{\qed}
\end{lemma}
\arxiv{\begin{proof}
  The full details are given in Appendix~\ref{app:proofsSR}.
\end{proof}}

\begin{lemma}[Equivalences $\forall$~\citeEquivForall]\label{lem:sr:equivforall}
  Let $U_1,\dots,U_n$ be a set of distinct (not equivalent) unit types and let
  $V_1,\dots,V_n$ be also a set of distinct unit types.
  \begin{enumerate}
  \item\label{ap:it:equivforall1} $\sui{n}\alpha_i\cdot
    U_i\equiv\suj{m}\beta_j\cdot V_j$ iff $\sui{n}\alpha_i\cdot\forall
    X.U_i\equiv\suj{m}\beta_j\cdot\forall X.V_j$.
  \item\label{ap:it:equivforall2} If $\sui{n}\alpha_i\cdot\forall
    X.U_i\equiv\suj{m}\beta_j\cdot V_j$ then $\forall V_j,\exists
    W_j~/~V_j\equiv\forall X.W_j$.
  \item\label{ap:it:equivforall3} If $T\equiv R$ then $T[A/X]\equiv R[A/X]$.
  \end{enumerate}
\end{lemma}
\arxiv{\begin{proof}
  The full details are given in Appendix~\ref{app:proofsSR}.
\end{proof}}
}

We follow Barendregt's proof of subject reduction for System
F~\cite{Barendregt92}, with the corrections first presented
at~\cite{stackexchange,ArrighiDiazcaroLMCS12}. First, we introduce a relation
between types, when these types are valid for the same term in the same context.

\begin{definition}\label{def:order} For any types $T, R$, and any context
  $\Gamma$ such that for some term $\ve{t}$, the sequent $\Gamma\vdash\ve t:T$ can be derived from the sequent $\Gamma\vdash\ve t:R$, without extra hypothesis, then
  \begin{enumerate}
  \item If $X\notin\FV{\Gamma}$, write $R\prec_{X,\Gamma} T$ if either:
    \begin{itemize}
    \item $R\equiv\sui{n}\alpha_i\cdot U_i$ and $T\equiv\sui{n}\alpha_i\cdot
      \forall X.U_i$,\quad or
    \item $R\equiv\sui{n}\alpha_i\cdot \forall X.U_i$ and $T\equiv
      \sui{n}\alpha_i\cdot U_i[A/X]$.
    \end{itemize}
  \item If $\V$ is a set of type variables such that
    $\V\cap\FV{\Gamma}=\emptyset$, we define $\preceq_{\V,\Gamma}$ inductively:
    \begin{itemize}
    \item If $R\prec_{X,\Gamma} T$, then $R\preceq_{\V \cup \{X\},\Gamma} T$.
    \item If $\V_1,\V_2\subseteq\V$, $S\preceq_{\V_1,\Gamma} R$ and
      $R\preceq_{\V_2,\Gamma} T$, then $S\preceq_{\V_1\cup\V_2,\Gamma} T$.
    \item If $R \equiv T$, then $R\preceq_{\V,\Gamma} T$.
    \end{itemize}
    Note that these relations only predicate on the types and the context, thus
    they hold for any term $\ve t$.
  \end{enumerate}
\end{definition}

\arxiv{\begin{example}
  Consider the following derivation.
  \[
    \prftree[r]{$\equiv$}
    {
      \prftree[r]{$\forall_{\vara{I}}$}
      {
	\prftree[r]{$\forall_{\varu{E}}$}
	{
	  \prftree[r]{$\forall_{\varu{I}}$}
	  {
	    \prftree[r]{$\equiv$}{\Gamma\vdash\ve t:T}
	    {\prfassumption{T\equiv \sui{n}\alpha_i\cdot U_i}}
	    {\Gamma\vdash\ve t:\sui{n}\alpha_i\cdot U_i}
	  }
	  {\prfassumption{\varu{X}\notin\FV{\Gamma}}}
	  {\Gamma\vdash\ve t:\sui{n}\alpha_i\cdot \forall\varu{X}.U_i}
	}
	{\Gamma\vdash\ve t:\sui{n}\alpha_i\cdot U_i[V/\varu{X}]}
      }
      {\prfassumption{\vara{Y}\notin\FV{\Gamma}}}
      {\Gamma\vdash\ve t:\sui{n}\alpha_i\cdot \forall\vara{Y}.U_i[V/\varu{X}]}
    }
    {\prfassumption{\sui{n}\alpha_i\cdot \forall\vara{Y}.U_i[V/\varu{X}]\equiv R}} 
    {\Gamma \vdash \ve{t}: R}
  \]
  Then $R\ssubt_{\{\varu{X},\vara{Y}\},\Gamma} T$.
\end{example}}
\conf{\begin{example}
  If there exists $\ve t$ such that $\Gamma\vdash\ve t:\sui{n}\alpha_i\cdot U_i$, then
  $$\sui{n}\alpha_i\cdot U_i\quad\ssubt_{\{\varu{X},\vara{Y}\},\Gamma}\quad\sui{n}\alpha_i\cdot \forall\vara{Y}.U_i[V/\varu{X}]$$
\end{example}}


\begin{lemma}\label{lem:sr:sorderhasnofv}
  For any unit type $U \not\equiv \forall X. V$, if $U \ssubt_{\V, \Gamma}
  \forall X. V$, then $X \notin \FV{\Gamma}$.
\end{lemma}
\arxiv{\begin{proof}
  By definition of $\ssubt$.
\end{proof}}

The following lemma states that if two arrow types are ordered, then they are
equivalent up to some substitution.

\begin{lemma}[Arrows comparison]\label{lem:sr:arrowscomp}
  $V \to R\ssubt_{\V,\Gamma} \forall\vec X.(U\to T) $, then $U\to T\equiv(V\to
  R)[\vec{A}/\vec{Y}]$, with $\vec Y\notin \FV{\Gamma}$.
\end{lemma}
\arxiv{\begin{proof}
  Let $(~\cdot~)^\circ$ be a map from types to types defined as follows,
  \begin{align*}
    X^\circ &= X \\
    (U\to T)^\circ &= U\to T \\
    (\forall X.T)^\circ &= T^\circ \\
    (\alpha\cdot T)^\circ &=\alpha\cdot T^\circ\\
    (T+R)^\circ &=T^\circ+R^\circ
  \end{align*}

  First we prove that for any types $V, U$, there exists $\vec A$ such that if
  $V \ssubt_{\V,\Gamma} \forall\vec X.U$, then $U^\circ\equiv V^\circ[\vec
  A/\vec X]$. Therefore, we have $U\to T\equiv(U\to T)^\circ\equiv(V\to
  R)^\circ[\vec A/\vec X]=(V\to R)[\vec A/\vec X]$. The full details of the
  proof are given in the Appendix~\ref{app:proofsSR}.
\end{proof}}

\arxiv{Five generation lemmas are required: two classical ones, for applications
(Lemma~\ref{lem:sr:app}) and abstractions (Lemma~\ref{lem:sr:abs}); and three
new ones for scalars (Lemma~\ref{lem:sr:scalars}), sums
(Lemma~\ref{lem:sr:sums}) and basis terms (Lemma~\ref{lem:sr:basevectors}).

\begin{lemma}[Scalars]\label{lem:sr:scalars}
  For any context $\Gamma$, term $\ve t$, type $T$, if $\pi = \Gamma\vdash
  \alpha\cdot\ve{t}: T$, there exist $R_1, \dots, R_n$, $\alpha_1, \dots,
  \alpha_n$ such that
  \begin{itemize}
  \item $T \equiv \sui{n}\alpha_i \cdot R_i$.
  \item $\pi_i = \Gamma \vdash \ve{t}: R_i$, with $size(\pi) > size(\pi_i)$, for
    $i \in \{1, \dots, n\}$.
  \item $\sui{n} \alpha_i = \alpha$.
  \end{itemize}
\end{lemma}
\begin{proof}
  By induction on the typing derivation. Full details are given in
  Appendix~\ref{app:proofsSR}.
\end{proof}

\begin{lemma}[Sums]\label{lem:sr:sums}
  If $\Gamma\vdash\ve t+\ve r:S$, there exist $R$, $T$ such that
  \begin{itemize}
  \item $S \equiv T + R$.
  \item $\Gamma\vdash\ve t: T$.
  \item $\Gamma\vdash\ve r: R$.
  \end{itemize}
\end{lemma}
\begin{proof}
  By induction on the typing derivation. Full details are given in
  Appendix~\ref{app:proofsSR}.
\end{proof}

\begin{lemma}[Application]\label{lem:sr:app}
  If $\Gamma\vdash(\ve t)~\ve r:T$, there exist $R_1, \dots, R_h$, $\mu_1,
  \dots, \mu_h$, $\V_1,\dots,\V_h$ such that $T \equiv \suk{h} \mu_k \cdot R_k$,
  $\suk{h} \mu_k = 1$ and for all $k \in \{1,\dots,h\}$
  \begin{itemize}
  \item $\Gamma\vdash\ve t: \sui{n_k}{\alpha_{(k,i)} \cdot\forall\vec{X}.(U\to
      T_{(k,i)})}$.
  \item $\Gamma\vdash\ve r: \suj{m_k}\beta_{(k,j)}\cdot
    U[\vec{A}_{(k,j)}/\vec{X}]$.
  \item $\sui{n_k}\suj{m_k} \alpha_{(k,i)}\times\beta_{(k,j)}\cdot
    {T_{(k,i)}[\vec{A}_{(k,j)}/\vec{X}]} \ssubt_{\V_k,\Gamma} R_k$.
  \end{itemize}
\end{lemma}
\begin{proof}
  By induction on the typing derivation. Full details are given in
  Appendix~\ref{app:proofsSR}.
\end{proof}

\begin{lemma}[Abstractions]\label{lem:sr:abs}
  If $\Gamma\vdash\lambda x.\ve t:T$, then there exist $T_1,\dots,T_n$,
  $R_1,\dots,R_n$, $U_1,\dots,U_n$, $\alpha_1,\dots,\alpha_n$, $\V_1,\dots,\V_n$
  such that $T \equiv \sui{n} \alpha_i \cdot T_i$, $\sui{n} \alpha_i = 1$ and
  for all $i \in \{1,\dots,n\}$,
  \begin{itemize}
  \item $\Gamma,x:U_i\vdash\ve t:R_i$.
  \item $U_i \to R_i \ssubt_{\V_i,\Gamma} T_i$.
  \end{itemize}
\end{lemma}
\begin{proof}
  By induction on the typing derivation. Full details are given in
  Appendix~\ref{app:proofsSR}.
\end{proof}

\begin{lemma}[Basis terms]\label{lem:sr:basevectors}
  For any context $\Gamma$, type $T$ and basis term $\ve{b}$, if
  $\Gamma\vdash\ve{b}: T$ there exist $U_1, \dots, U_n$, $\alpha_1, \dots,
  \alpha_n$ such that
  \begin{itemize}
  \item $T \equiv \sui{n} \alpha_i \cdot U_i$.
  \item $\Gamma\vdash\ve{b}: U_i$, for $i \in \{1,\dots,n\}$.
  \item $\sui{n} \alpha_i = 1$.
  \end{itemize}
\end{lemma}
\begin{proof}
  By induction on the typing derivation. Full details are given in
  Appendix~\ref{app:proofsSR}.
\end{proof}
}
\conf{The relation between types just defined is taken into account for the generation lemmas. We left the technical details for the arXiv'ed long version~\cite{paper:arxiv}.}

Substitution lemma is standard.

\begin{lemma}[Substitution lemma]\label{lem:sr:substitution}
  For any term ${\ve t}$, basis term $\ve b$, term variable $x$, context
  $\Gamma$, types $T$, $U$, type variable $X$ and type $A$, where $A$ is a unit
  type if $X$ is a unit variable, otherwise $A$ is a general type, we have,
  \begin{enumerate}
  \item \label{ap:it:substitutionTypes} if $\Gamma\vdash\ve{t}: T$, then
    $\Gamma[A/X]\vdash\ve{t}: T[A/X]$;
  \item \label{ap:it:substitutionTerms} if $\Gamma,x:U\vdash\ve t:T$ and
    $\Gamma\vdash\ve b:U$, then $\Gamma\vdash\ve t[\ve b/x]: T$.
  \end{enumerate}
\end{lemma}
\arxiv{\begin{proof}
  Both items are proven by induction on the typing derivation. Full details are
  given in Appendix~\ref{app:proofsSR}.
\end{proof}}

We extend the equivalence between types as an equivalence between contexts in a
natural way: The equivalence between contexts $\Gamma \equiv \Delta$ is defined
by $x:U \in \Gamma$ if and only if there exists $x:V \in \Delta$ such that $U
\equiv V$.

\begin{theorem}[Subject Reduction]\label{thm:sr}
  For any terms $\ve{t}, \ve{t}'$, any context $\Gamma$ and any type $T$, if
  $\ve{t} \to \ve{t}'$ and $\Gamma \vdash \ve{t}: T$, then $\Gamma \vdash
  \ve{t}': T$.
\end{theorem}
\arxiv{\begin{proof}
  By induction on the rewrite relation. Full details are given in
  Appendix~\ref{app:proofsSR}.
\end{proof}}

\section{Other properties}\label{ch:other-properties}
In this section we present additional properties that are satisfied by $\lvecr$:
progress, strong normalisation, and a characterisation property showing that the sum of all the
components of a vector, which we call weight, of a type is the weight of the
value obtained after reduction.
\conf{We give enough details for reproducing all the proofs. The full detailed proofs are given in the 51-pages long arXiv'ed version at~\cite{paper:arxiv}.}
\arxiv{The proofs are given in the Appendix~\ref{app:proofsOP}.}

Let $\mathbb{V} = \left\{\sui{n} \alpha_i \cdot \lambda x_i.\ve{t}_i +
  \sum^{m}_{j=n+1} \lambda x_j.\ve{t}_j \mid \forall i, j, \lambda x_i.\ve{t}_i
  \neq \lambda x_j.\ve{t}_j\right\}$ be the set of values in our calculus, and
we write $\mathsf{NF}$ as the set of terms in normal form (that is, terms that
cannot be reduced any further). The following theorem relates those two sets.

\begin{theorem}[Progress]\label{thm:progress}
  If $\vdash \ve{t}: T$ and $\ve{t} \in \mathsf{NF}$, then $\ve{t} \in
  \mathbb{V}$.
\end{theorem}
\arxiv{\begin{proof}
  By induction on $\ve{t}$. Full details are given in
  Appendix~\ref{app:proofsOP}.
\end{proof}}

\begin{theorem}[Strong Normalisation]\label{thm:sn}
If $\Gamma \vdash \ve{t}: T$ is a valid sequent, then $\ve{t}$ is strongly normalising.
\end{theorem}
\begin{proof}
 The proof is by showing that every typed term in \lvecr\ is also typed in \lvec. The full details are given in the Appendix~\ref{app:proofsOP}.
\end{proof}

As previously discussed, the objective of the system is to be able to model
vector spaces (or, more precisely, weak modules). In this context, we know that
the basis terms represent base vectors, while general terms represent any
vector. From here, it follows that if $\ve{v} = \alpha \cdot \ve{b}_1 + \beta
\cdot \ve{b}_2$, then $\ve{b}_1$ represents the vector
$\left(\begin{smallmatrix}1\\0\end{smallmatrix}\right)$, $\ve{b}_2$ represents
the vector $\left(\begin{smallmatrix}0\\1\end{smallmatrix}\right)$, and $\ve{v}$
represents the vector
$\left(\begin{smallmatrix}\alpha\\\beta\end{smallmatrix}\right) = \alpha\cdot
\left(\begin{smallmatrix}1\\0\end{smallmatrix}\right) + \beta \cdot
\left(\begin{smallmatrix}0\\1\end{smallmatrix}\right)$. Therefore, the weight of
$\ve{v}$ should be $\alpha + \beta$, since that is effectively the weight of
$\left(\begin{smallmatrix}\alpha\\\beta\end{smallmatrix}\right)$.

This is analogous for types: the unit types represent base vectors (which is why
they type basis terms), and the general types represent any vector.

We proceed then to formalise the concept of weight of types and terms. 
First we define the weight of types (Definition~\ref{def:wp:weighttypes}), then the weight of values (Definition~\ref{def:wp:weightterms}), and, finally,  
we can define the weight of a term as the weight of its type, after proving that if a typed term reduces to a value, then the weight of the value and of the type coincides (Theorem~\ref{thm:wp:weightpreserv}).

\begin{definition}[Weight of types]\label{def:wp:weighttypes}
  We define the relation $\tnorm{\bullet}: \text{Type} \to \text{Scalar}$
  inductively as follows:
  \[
    \begin{array}{r@{\,}l@{\hspace{1cm}}r@{\,}l@{\hspace{1cm}}r@{\,}l}
      \tnorm{U} &= 1
      &
      \tnorm{\alpha \cdot T} &= \alpha \cdot \tnorm{T}
      &
      \tnorm{T + R} &= \tnorm{T} + \tnorm{R}
    \end{array}
  \]
\end{definition}

\begin{example}
  Consider the type $\sui{n} \alpha_i \cdot U_i$, then
  \[
    \tnorm{\sui{n} \alpha_i \cdot U_i} = \sui{n} \alpha_i \cdot \tnorm{U_i}=
    \sui{n} \alpha_i
  \]
\end{example}
\begin{definition}[Weight of values]\label{def:wp:weightterms}
  We define the relation $\tnorm{\bullet}: \text{Term} \to \text{Scalar}$
  inductively as follows:
  \[
    \begin{array}{r@{\,}l@{\hspace{1cm}}r@{\,}l@{\hspace{1cm}}r@{\,}l}
      \tnorm{\ve{b}} &= 1
      &
      \tnorm{\alpha \cdot \ve{t}} &= \alpha \cdot \tnorm{\ve{t}}
      &
      \tnorm{\ve{t} + \ve{r}} &= \tnorm{\ve{t}} + \tnorm{\ve{r}}
    \end{array}
  \]
\end{definition}
\begin{example}
  Consider the term $\sui{n} \alpha_i \cdot \lambda x_i.\ve{t}_i$, then
  \[
    \tnorm{\sui{n} \alpha_i \cdot \lambda x_i.\ve{t}_i} = \sui{n} \alpha_i \cdot
    \tnorm{\lambda x_i.\ve{t}_i}= \sui{n} \alpha_i
  \]
\end{example}

\begin{lemma}\label{lem:wp:weightequiv}
  If $T \equiv R$, then $\tnorm{T} = \tnorm{R}$.
\end{lemma}
\arxiv{\begin{proof}
  We prove the lemma holds for every definition of $\equiv$. Full details are
  given in Appendix~\ref{app:proofsOP}.
\end{proof}}

\begin{lemma}\label{lem:wp:weightofvalues}
  If $\ve{v}\in \mathbb{V}$, and $\vdash \ve{v}: T$, then $\tnorm{T}
  \equiv \tnorm{\ve{v}}$.
\end{lemma}
\arxiv{\begin{proof}
  By induction on $n$. Full details are given in Appendix~\ref{app:proofsOP}.
\end{proof}}

Finally, the weight of an arbitrary term can be defined as the weight of its
type, thanks to the following theorem.

\begin{theorem}[Weight Preservation]\label{thm:wp:weightpreserv}
  If $\vdash \ve{t}: T$ and $\ve{t} \to^{*} \ve{v}$, then $\tnorm{T} =
  \tnorm{\ve{v}}$.
\end{theorem}
\begin{proof}
  Since $\ve{t} \to^{*} \ve{v}$, by Theorem~\ref{thm:progress}, $\ve{v} = \sui{n}
  \alpha_i \cdot \lambda x_i.\ve{t}_i + \sum^{m}_{j=n+1} \lambda x_j.\ve{t}_j$,
  where $\lambda x_i.\ve{t}_i \neq \lambda x_j.\ve{t}_j$ for all $i \in \{1,
  \dots, n\}$, $j \in \{n+1, \dots m\}$. Also, by Theorem~\ref{thm:sr}, we know then
  that $\vdash \ve{v}: T$. Finally, by Lemma~\ref{lem:wp:weightofvalues}, we know
  that $\tnorm{T} = \tnorm{\ve{v}}$.
\end{proof}

\section{Conclusion}\label{ch:conclusion}
We have revisited $\lvec$ redefining it in a careful way, proving that the
modified version satisfies the standard formulation of the Subject Reduction
property (Theorem~\ref{thm:sr}).
It is worth mentioning that the design choices we made
are not necessarily the only possibility. Indeed, one of the first approaches we
considered involved keeping most of the typing rules as in the original system,
and adding subtyping.
In the end, we realized that the property could be satisfied in a simpler and more elegant way by modifying
the typing rules. The summary of the changes made to the original system is:
\begin{itemize}
\item We added the $S$ rule, that deals with superposition of types of a single
  term.
\item We added the $1_E$ rule, to allow the removal of the scalar if said scalar
  is equal to 1.
\item We removed the term $\ve{0}$, which proved to be
  undesirable~\cite{DiazcaroGuillermoMiquelValironLICS19}.
\end{itemize}

In addition, we showed that the obtained calculus is still strongly normalising
(Theorem~\ref{thm:sn}), by proving that the typable terms in the modified
version, are typable in the original system (which has been proved to be
strongly normalising as well~\cite{vectorial}). We also provided a proof of the
progress property (Theorem~\ref{thm:progress}), which allowed us to characterise
the terms that cannot be reduced any further. This enabled us to formalize the
concept of weight of types and terms, and to prove that terms had the same
weight as their types (Theorem~\ref{thm:wp:weightpreserv}).

We stand by this modified version of \lvec, which we think provides the right framework in the quest for the
quantum-controlled lambda calculus.

\bibliographystyle{splncs04}
\bibliography{bibliography}

\arxiv{\newpage}
\appendix

\arxiv{
\section{Omitted proofs in Section~\ref{ch:sr}}\label{app:proofsSR}
\arxiv{\xrecap{Lemma}{Characterisation of types~\citeCharTypes}{lem:sr:typecharact} {
  For any type $T$, there exist $n,m\in\mathbb{N}$, $\alpha_1,\dots,\alpha_n$,
  $\beta_1,\dots,\beta_m\in\Sc$, distinct unit types $U_1,\dots,U_n$ and
  distinct general variables $\vara{X}_1,\dots,\vara{X}_m$ such that \(
  T\equiv\sui{n}\alpha_i\cdot U_i+\suj{m}\beta_j\cdot\vara{X}_j \). }
\begin{proof}
  Structural induction on $T$.
  \begin{itemize}
  \item Let $T=U$, then take $\alpha=\beta=1$, $n=1$ and $m=0$, and so
    $T\equiv\sui{1}1\cdot U=1\cdot U$.
  \item Let $T=\alpha\cdot T'$, then by the induction hypothesis
    $T'\equiv\sui{n}\alpha_i\cdot U_i+\suj{m}\beta_j\cdot\vara{X}_j$, so
    $T=\alpha\cdot T'\equiv\alpha\cdot (\sui{n}\alpha_i\cdot
    U_i+\suj{m}\beta_j\cdot\vara{X}_j)\equiv\sui{n}(\alpha\times\alpha_i)\cdot
    U_i+\suj{m}(\alpha\times\beta_j)\cdot\vara{X}_j$.
  \item Let $T=R+S$, then by the induction hypothesis
    $R\equiv\sui{n}\alpha_i\cdot U_i+\suj{m}\beta_j\cdot\vara{X}_j$ and
    $S\equiv\sui{n'}\alpha'_i\cdot U'_i+\suj{m'}\beta'_j\cdot\vara{X'}_j$, so
    $T=R+S\equiv\sui{n}\alpha_i\cdot U_i+\sui{n'}\alpha'_i\cdot
    U'_i+\suj{m}\beta_j\cdot\vara{X}_j+\suj{m'}\beta'_j\cdot\vara{X'}_j$. If the
    $U_i$ and the $U'_i$ are all different each other, we have finished, in
    other case, if $U_k=U'_h$, notice that $\alpha_k\cdot U_k+\alpha'_h\cdot
    U'_h\equiv (\alpha_k+\alpha'_h)\cdot U_k$.
  \item Let $T=\vara X$, then take $\alpha=\beta=1$, $m=1$ and $n=0$, and so
    $T\equiv\suj{1} 1\cdot\vara{X}\equiv 1\cdot\vara X$. \qed
  \end{itemize}
\end{proof}}

\begin{definition}
  Let $F$ be an algebraic context with $n$ holes. Let $\vec U = U_1,\ldots,U_n$
  be a list of $n$ unit types. If $U$ is a unit type, we write $\bar U$ for the
  set of unit types equivalent to $U$:
  \[
    \bar U := \{ V ~|~ V \textrm{ is unit and } V \equiv U \}.
  \]
  The {\em context vector} $v_F(\vec U)$ associated with the context $F$ and the
  unit types $\vec U$ is partial map from the set $\mathcal S = \{ \bar U \}$ to
  scalars. It is inductively defined as follows: $v_{\alpha\cdot F}(\vec U) :=
  \alpha v_F(\vec U)$, $v_{F + G}(\vec U) := v_F(\vec U) + v_G(\vec U)$, and
  finally $v_{[-_i]}(\vec U) := \{\bar U_i \mapsto 1\}$. The sum is defined on
  these partial map as follows:
  \[
    (f + g)(\vec U) = \left\{
      \begin{array}{ll}
        f(\vec U) + g(\vec U)& \textrm{if both are defined}
        \\
        f(\vec U)& \textrm{if $f(\vec U)$ is defined but not $g(\vec U)$}
        \\
        g(\vec U)& \textrm{if $g(\vec U)$ is defined but not $f(\vec U)$}
        \\
        \textrm{is not defined} & \textrm{if neither $f(\vec U)$ nor
                                  $g(\vec U)$ is defined.}
      \end{array}
    \right.
  \]
  Scalar multiplication is defined as follows:
  \[
    (\alpha f)(\vec U) = \left\{
      \begin{array}{ll}
        \alpha (f(\vec U))& \textrm{if $f(\vec U)$ is defined}
        \\
        \textrm{is not defined} & \textrm{if $f(\vec U)$ is not defined.}
      \end{array}
    \right.
  \]
\end{definition}

\begin{lemma}\label{lem:equivdistinctscalarsaux}
  Let $F$ and $G$ be two algebraic contexts with respectively $n$ and $m$ holes.
  Let $\vec U$ be a list of $n$ unit types, and $\vec V$ be a list of $m$ unit
  types. Then $F(\vec U) \equiv G(\vec V)$ implies $v_F(\vec U) = v_G(\vec V)$.
\end{lemma}

\begin{proof}
  The derivation of $F(\vec U) \equiv F(\vec V)$ essentially consists
  in a sequence of the elementary rules (or congruence thereof) in
  Figure~\ref{fig:typeequiv} composed with transitivity:
  \[
    F(\vec U) = W_1\equiv W_2 \equiv \cdots \equiv W_k = G(\vec V).
  \]
  We prove the result by induction on $k$.
  \begin{itemize}
  \item Case $k=1$. Then $F(\vec U)$ is syntactically equal to
    $G(\vec V)$: we are done.
  \item Suppose that the result is true for sequences of size $k$, and
    let
    \[
      F(\vec U) = W_1\equiv W_2 \equiv \cdots \equiv W_k
      \equiv W_{k+1} = G(\vec V).
    \]
    Let us concentrate on the first step $F(\vec U) \equiv W_2$: it is an
    elementary step from Figure~\ref{fig:typeequiv}. By structural induction on
    the proof of $F(\vec U) \equiv W_2$ (which only uses congruence and
    elementary steps, and not transitivity), we can show that $W_2$ is of the
    form $F'(\vec U')$ where $v_F(\vec U) = v_{F'}(\vec U')$. We are now in
    power of applying the induction hypothesis, because the sequence of
    elementary rewrites from $F'(\vec U')$ to $G(\vec V)$ is of size $k$.
    Therefore $v_{F'}(\vec U') = v_G(\vec V)$. We can then conclude that
    $v_F(\vec U) = v_G(\vec V)$.
  \end{itemize}
  This conclude the proof of the lemma.
\end{proof}

\xrecap{Lemma}{Equivalence between sums of distinct elements (up to
  $\equiv$)}{lem:sr:equivdistinctscalars} Let $U_1,\dots,U_n$ be a set of
distinct (not equivalent) unit types, and let $V_1,\dots,V_m$ be also a set
distinct unit types. If $\sui{n}\alpha_i\cdot U_i\equiv\suj{m}\beta_j\cdot V_j$,
then $m=n$ and there exists a permutation $p$ of $n$ such that $\forall i$,
$\alpha_i=\beta_{p(i)}$ and $U_i\equiv V_{p(i)}$.

\begin{proof}
  Let $S = \sui{n}\alpha_i\cdot U_i$ and $T = \suj{m}\beta_j\cdot V_j$. Both $S$
  and $T$ can be respectively written as $F(\vec U)$ and $G(\vec V)$. Using
  Lemma~\ref{lem:equivdistinctscalarsaux}, we conclude that $v_F(\vec U) =
  v_G(\vec V)$. Since all $U_i$'s are pairwise non-equivalent, the $\bar U_i$'s
  are pairwise distinct.
  \[
    v_F(\vec U) = \{ \bar U_i \mapsto \alpha_i~|~i=1\ldots n\}.
  \]
  Similarly, the $\bar V_j$'s are pairwise disjoint, and
  \[
    v_G(\vec G) = \{ \bar V_j \mapsto \beta_j~|~i=1\ldots m\}.
  \]
  We obtain the desired result because these two partial maps are supposed to be
  equal. Indeed, this implies:
  \begin{itemize}
  \item $m=n$ because the domains are equal (so they should have the same size)
  \item Again using the fact that the domains are equal, the sets $\{\bar U_i\}$
    and $\{\bar V_j\}$ are equal: this means there exists a permutation $p$ of
    $n$ such that $\forall i$, $\bar U_i= \bar V_{p(i)}$, meaning $U_i\equiv
    V_{p(i)}$.
  \item Because the partial maps are equal, the images of a given element $\bar
    U_i = \bar V_{p(i)}$ under $v_F$ and $v_G$ are in fact the same: we
    therefore have $\alpha_i=\beta_{p(i)}$.
  \end{itemize}
  And this closes the proof of the lemma.
\end{proof}

\arxiv{\xrecap{Lemma}{Equivalences $\forall$~\citeEquivForall}{lem:sr:equivforall} {
  Let $U_1,\dots,U_n$ be a set of distinct (not equivalent) unit types and let
  $V_1,\dots,V_n$ be also a set of distinct unit types.
  \begin{enumerate}
  \item $\sui{n}\alpha_i\cdot
    U_i\equiv\suj{m}\beta_j\cdot V_j$ iff $\sui{n}\alpha_i\cdot\forall
    X.U_i\equiv\suj{m}\beta_j\cdot\forall X.V_j$.
  \item If $\sui{n}\alpha_i\cdot\forall
    X.U_i\equiv\suj{m}\beta_j\cdot V_j$ then $\forall V_j,\exists
    W_j~/~V_j\equiv\forall X.W_j$.
  \item If $T\equiv R$ then $T[A/X]\equiv R[A/X]$.
  \end{enumerate}
}
\begin{proof}
  Item (1) From Lemma~\ref{lem:sr:equivdistinctscalars}, $m=n$, and without loss
  of generality, for all $i$, $\alpha_i=\beta_i$ and $U_i=V_i$ in the
  left-to-right direction, $\forall X.U_i=\forall X.V_i$ in the right-to-left
  direction. In both cases we easily conclude.

  \medskip
  \noindent
  Item (2) is similar.

  \medskip
  \noindent
  Item (3) is a straightforward induction on the equivalence $T\equiv R$.
\end{proof}
}

\xrecap{Lemma}{Arrows comparison}{lem:sr:arrowscomp} { $V \to
  R\ssubt_{\V,\Gamma} \forall\vec X.(U\to T) $, then $U\to T\equiv(V\to
  R)[\vec{A}/\vec{Y}]$, with $\vec Y\notin \FV{\Gamma}$. }
\begin{proof}
  Let $(~\cdot~)^\circ$ be a map from types to types defined as follows,
  \begin{align*}
    X^\circ &= X \\
    (U\to T)^\circ &= U\to T \\
    (\forall X.T)^\circ &= T^\circ \\
    (\alpha\cdot T)^\circ &=\alpha\cdot T^\circ\\
    (T+R)^\circ &=T^\circ+R^\circ
  \end{align*}

  We need three intermediate results:
  \begin{enumerate}
  \item If $T\equiv R$, then $T^\circ\equiv R^\circ$.
  \item For any types $U, A$, there exists $B$ such that
    $(U[A/X])^\circ=U^\circ[B/X]$.
  \item For any types $V, U$, there exists $\vec A$ such that if $V
    \ssubt_{\V,\Gamma} \forall\vec X.U$, then $U^\circ\equiv V^\circ[\vec A/\vec
    X]$.
  \end{enumerate}
  {\textit{Proofs.}}
  \begin{enumerate}
  \item Induction on the equivalence rules. We only give the basic cases since
    the inductive step, given by the context where the equivalence is applied,
    is trivial.
    \begin{itemize}
    \item $(1\cdot T)^\circ=1\cdot T^\circ\equiv T^\circ$.
    \item $(\alpha\cdot(\beta\cdot T))^\circ=\alpha\cdot(\beta\cdot
      T^\circ)\equiv(\alpha\times\beta)\cdot T^\circ=((\alpha\times\beta)\cdot
      T)^\circ$.
    \item $(\alpha\cdot T+\alpha\cdot R)^\circ=\alpha\cdot T^\circ+\alpha\cdot
      R^\circ\equiv\alpha\cdot(T^\circ+R^\circ)=(\alpha\cdot(T+R))^\circ$.
    \item $(\alpha\cdot T+\beta\cdot T)^\circ=\alpha\cdot T^\circ+\beta\cdot
      T^\circ\equiv(\alpha+\beta)\cdot T^\circ=((\alpha+\beta)\cdot T)^\circ$.
    \item $(T+R)^\circ=T^\circ+R^\circ\equiv R^\circ+T^\circ=(R+T)^\circ$.
    \item $(T+(R+S))^\circ=T^\circ+(R^\circ+S^\circ)\equiv
      (T^\circ+R^\circ)+S^\circ=((T+R)+S)^\circ$.
    \end{itemize}
  \item Structural induction on $U$.
    \begin{itemize}
    \item $U=\varu X$. Then $(\varu X[V/\varu X])^\circ=V^\circ=\varu
      X[V^\circ/\varu X]=\varu X^\circ[V^\circ/\varu X]$.
    \item $U=\varu Y$. Then $(\varu Y[A/X])^\circ=\varu Y=\varu Y^\circ[A/X]$.
    \item $U=V\to T$. Then $((V\to T)[A/X])^\circ=(V[A/X]\to
      T[A/X])^\circ=V[A/X]\to T[A/X]=(V\to T)[A/X]=(V\to T)^\circ[A/X]$.
    \item $U=\forall Y.V$. Then $((\forall Y.V)[A/X])^\circ=(\forall
      Y.V[A/X])^\circ=(V[A/X])^\circ$, which by the induction hypothesis is
      equivalent to $V^\circ[B/X]=(\forall Y.V)^\circ[B/X]$.
    \end{itemize}
  \item It suffices to show this for $V \prec_{X,\Gamma} \forall\vec X.U$.
    Cases:
    \begin{itemize}
    \item $\forall\vec X.U\equiv\forall Y.V$, then notice that $(\forall\vec
      X.U)^\circ \equiv_{(1)}(\forall Y.V)^\circ=V^\circ$.
    \item $V\equiv\forall Y.W$ and $\forall\vec X.U\equiv W[A/X]$, then

      $(\forall\vec X.U)^\circ\equiv_{(1)}(W[A/X])^\circ\equiv_{(2)}
      W^\circ[B/X]=(\forall Y.W)^\circ[B/X]\equiv_{(1)}V^\circ[B/X]$.
    \end{itemize}
  \end{enumerate}
  Proof of the lemma. $U\to T\equiv(U\to T)^\circ$, by the intermediate result
  3, this is equivalent to $(V\to R)^\circ[\vec A/\vec X]=(V\to R)[\vec A/\vec
  X]$.
\end{proof}

\xrecap{Lemma}{Scalars}{lem:sr:scalars}{ For any context $\Gamma$, term $\ve t$,
  type $T$, if $\pi = \Gamma\vdash \alpha\cdot\ve{t}: T$, there exist $R_1,
  \dots, R_n$, $\alpha_1, \dots, \alpha_n$ such that
  \begin{itemize}
  \item $T \equiv \sui{n}\alpha_i \cdot R_i$.
  \item $\pi_i = \Gamma \vdash \ve{t}: R_i$, with $size(\pi) > size(\pi_i)$, for
    $i \in \{1, \dots, n\}$.
  \item $\sui{n} \alpha_i = \alpha$.
  \end{itemize}
}
\begin{proof}
  By induction on the typing derivation. \inductioncase{Case $S$}
  \[
    \prftree[r]{$S$} {\Gamma \vdash \ve{t}: T_i} {\prfassumption{\forall i \in
        \{1,\dots,n\}}}
    {\Gamma \vdash \left(\sui{n} \alpha_i\right) \cdot \ve{t}: \sui{n} \alpha_i
      \cdot T_i}
  \]
  Trivial case. \inductioncase{Case $\equiv$}
  \[
    \prftree[r]{$\equiv$} {\pi' = \Gamma \vdash \alpha \cdot \ve{t}: T} {T
      \equiv R} {\pi = \Gamma \vdash \alpha \cdot \ve{t}: R}
  \]
  By the induction hypothesis there exist $S_1,\dots,S_n$,
  $\alpha_1,\dots,\alpha_n$ such that
  \begin{itemize}
  \item $T \equiv R \equiv \sui{n} \alpha_i \cdot S_i$.
  \item $\pi_i = \Gamma \vdash \ve{t}: S_i$, with $size(\pi') > size(\pi_i)$,
    for $i \in \{1,\dots,n\}$.
  \item $\sui{n} \alpha_i = \alpha$.
  \end{itemize}
  It is easy to see that $size(\pi) > size(\pi')$, so the lemma holds.
  \inductioncase{Case $1_E$}
  \[
    \prftree[r]{$1_E$} {\pi = \Gamma \vdash 1\cdot (\alpha \cdot \ve{t}): T}
    {\Gamma \vdash \alpha \cdot \ve{t}: T}
  \]
  By induction hypothesis, there exist $R_1,\dots,R_m$, $\beta_1,\dots,\beta_m$
  such that
  \begin{itemize}
  \item ${T \equiv \suj{m} \beta_j \cdot R_j}$.
  \item $\pi_j = \Gamma \vdash \alpha \cdot \ve{t}: R_j$ with $size(\pi) >
    size(\pi_j)$ for $j = \{1, \dots, m\}$.
  \item $\suj{m} \beta_j = 1$.
  \end{itemize}
  Since $size(\pi) > size(\pi_j)$, then by applying the induction hypothesis
  again for all $j = \{1, \dots, m\}$, we have that there exist
  $S_{(j,1)},\dots,S_{(j,n_j)}$, $\alpha_{(j,1)},\dots,\alpha_{(j,n_j)}$ such
  that
  \begin{itemize}
  \item $R_j \equiv \sui{n_j} \alpha_{(j,i)} \cdot S_{(j,i)}$.
  \item $\pi_{(j,i)} = \Gamma \vdash \ve{t}: S_{(j,i)}$ with $size(\pi_j) >
    size(\pi_{(j,i)})$ for $i \in \{1, \dots, n_j\}$.
  \item $\sui{n_j} \alpha_{(j,i)} = \alpha$.
  \end{itemize}
  Given that $\Gamma \vdash \alpha \cdot \ve{t}: T$, then
  \[
    T \equiv \suj{m} \beta_j \cdot R_j \equiv \suj{m} \beta_j \cdot \sui{n}
    \alpha_{(j,i)} \cdot S_{(j,i)} \equiv \suj{m}\sui{n} (\beta_j \times
    \alpha_{(j,i)}) \cdot S_{(j,i)}
  \]
  Finally, we must prove that $\suj{m}\sui{n} (\beta_j \times \alpha_{(j,i)}) =
  \alpha$,
  \[
    \suj{m}\sui{n} (\beta_j \times \alpha_{(j,i)}) = \suj{m} \beta_j \cdot
    \underbrace{\sui{n} \alpha_{(j,i)}}_{=~\alpha} = \suj{m} \beta_j \cdot
    \alpha = \alpha \cdot \underbrace{\suj{m} \beta_j}_{=~1} = \alpha
  \]
  \inductioncase{Case $\forall_I$}
  \[
    \prftree[r]{$\forall_I$} {\pi = \Gamma \vdash \alpha \cdot \ve{t}:
      \sui{n}\alpha_i \cdot U_i} {X \notin \FV{\Gamma}} {\pi' = \Gamma \vdash
      \alpha \cdot \ve{t}: \sui{n}\alpha_i \cdot \forall X. U_i}
  \]
  By the induction hypothesis there exist $R_1,\dots,R_m$, $\mu_1,\dots,\mu_m$
  such that
  \begin{itemize}
  \item $\sui{n}\alpha_i\cdot U_i\equiv \suj{m} \mu_j \cdot R_j$.
  \item $\pi_j = \Gamma \vdash \ve{t}: R_j$, with $size(\pi) > size(\pi_j)$, for
    $j \in \{1, \dots, m\}$.
  \item $\suj{m} \mu_j = \alpha$.
  \end{itemize}
  By applying Lemma~\ref{lem:sr:typecharact} for all $j \in \{1, \dots, m\}$,
  and since $\sui{n}\alpha_i\cdot U_i$ does not have any general variable
  $\vara{X}$,
  then $R_j \equiv \suk{h_j}\beta_{(j,k)}\cdot V_{(j,k)}$.\\
  Hence $\sui{n}\alpha_i\cdot U_i\equiv \suj{m}\mu_j \cdot \suk{h_j}\beta_{(j,k)}\cdot V_{(j,k)}$.\\
  Without loss of generality, assuming all unit types are distinct (not
  equivalent), then by Lemma~\ref{lem:sr:equivforall},
  \[
    \sui{n-1}\alpha_i\cdot \forall{X}.U_n \equiv \suj{m}\mu_j \cdot
    \underbrace{\suk{h_j}\beta_{(j,k)}\cdot \forall X.V_{(j,k)}}_{\equiv~R'_j}
  \]
  We must prove that for all $j \in \{1,\dots,m\}$, $\pi'_j = \Gamma \vdash
  \ve{t}: R'_j$ and that $size(\pi') > size(\pi'_j)$. By applying the
  $\forall_I$ rule, we have
  \[
    \prftree[r]{$\forall_I$} {\Gamma \vdash \ve{t}: R_j} {\prfassumption{X
        \notin \FV{\Gamma}}} {\pi'_j = \Gamma \vdash \ve{t}: R'_j}
  \]
  And notice that using the $S$ rule, obtain
  \[
    \prftree[r]{$\equiv$} {\prftree[r]{$S$} {\pi'_j = \Gamma \vdash \ve{t}:
        R'_j} {\forall j \in \{1,\dots,m\}} {\Gamma \vdash \alpha \cdot \ve{t}:
        \suj{m} \mu_j \cdot R'_j}} {\prfassumption{\sui{n}\alpha_i \cdot \forall
        X. U_i \equiv \suj{m} \mu_j \cdot R'_j}} {\pi' = \Gamma \vdash \alpha
      \cdot \ve{t}: \sui{n}\alpha_i \cdot \forall X. U_i}
  \]
  So for all $j \in \{1,\dots,m\}$, $size(\pi') > size(\pi'_j)$.
  \inductioncase{Case $\forall_E$}
  \[
    \prftree[r]{$\forall_E$} {\pi = \Gamma \vdash \alpha \cdot \ve{t}:
      \sui{n}\alpha_i \cdot \forall X. U_i} {\pi' = \Gamma \vdash \alpha \cdot
      \ve{t}: \sui{n}\alpha_i \cdot U_i[A/X]}
  \]
  By the induction hypothesis there exist $R_1, \dots, R_m$, $\mu_1, \dots,
  \mu_m$ such that
  \begin{itemize}
  \item $\sui{n} \alpha_i \cdot \forall X. U_i \equiv \suj{m} \mu_j \cdot R_j$.
  \item $\pi_j = \Gamma \vdash \ve{t}: R_j$, with $size(\pi) > size(\pi_j)$, for
    $j \in \{1, \dots, m\}$.
  \item $\suj{m} \mu_j = \alpha$.
  \end{itemize}
  By applying Lemma~\ref{lem:sr:typecharact} for all $j \in \{1, \dots, m\}$,
  and since $\sui{n}\alpha_i \cdot \forall X. U_i$
  does not have any general variable $\vara{X}$, then $R_j \equiv \suk{h_j}\beta_{(j,k)}\cdot V_{(j,k)}$.\\
  Hence $\sui{n}\alpha_i \cdot \forall X. U_i \equiv \suj{m}\mu_j \cdot \suk{h_j}\beta_{(j,k)}\cdot V_{(j,k)}$.\\
  Without loss of generality, we assume that all unit types present at both
  sides of the equivalence are distinct, then by
  Lemma~\ref{lem:sr:equivdistinctscalars}, for all $j \in \{1, \dots, m\}$, $k
  \in \{1,\dots,h_j\}$, there exists $V'_{(j,k)}$ such that $V_{(j,k)} \equiv
  \forall X.V'_{(j,k)}$. Then,
  \[
    \sui{n}\alpha_i\cdot \forall{X}.U_i \equiv \suj{m}\mu_j \cdot
    \underbrace{\suk{h_j}\beta_{(j,k)}\cdot \forall X.V'_{(j,k)}}_{\equiv~R_j}
  \]
  By the same lemma, we have that
  \[
    \sui{n}\alpha_i\cdot U_i[A/X] \equiv \suj{m}\mu_j \cdot
    \underbrace{\suk{h_j}\beta_{(j,k)}\cdot V'_{(j,k)}[A/X]}_{\equiv~R'_j}
  \]
  We must prove that for all $j \in \{1,\dots,m\}$, $\pi'_j = \Gamma \vdash
  \ve{t}: R'_j$ and that $size(\pi') > size(\pi'_j)$. By applying the
  $\forall_E$ rule, we have
  \[
    \prftree[r]{$\forall_E$} {\Gamma \vdash \ve{t}: R_j} {\pi'_j = \Gamma \vdash
      \ve{t}: R'_j}
  \]
  And notice that using the $S$ rule, obtain
  \[
    \prftree[r]{$\equiv$} {\prftree[r]{$S$} {\pi'_j = \Gamma \vdash \ve{t}:
        R'_j} {\forall j \in \{1,\dots,m\}} {\Gamma \vdash \alpha \cdot \ve{t}:
        \suj{m} \mu_j \cdot R'_j}} {\prfassumption{\sui{n}\alpha_i \cdot
        U_i[A/X] \equiv \suj{m} \mu_j \cdot R'_j}} {\pi' = \Gamma \vdash \alpha
      \cdot \ve{t}: \sui{n}\alpha_i \cdot U_i[A/X]}
  \]
  So for all $j \in \{1,\dots,m\}$, $size(\pi') > size(\pi'_j)$.
\end{proof}

\xrecap{Lemma}{Sums}{lem:sr:sums}{ If $\Gamma\vdash\ve t+\ve r:S$, there exist
  $R$, $T$ such that
  \begin{itemize}
  \item $S \equiv T + R$.
  \item $\Gamma\vdash\ve t: T$.
  \item $\Gamma\vdash\ve r: R$.
  \end{itemize}
}
\begin{proof}
  By induction on the typing derivation. \inductioncase{Case $+_I$}
  \[
    \prftree[r]{$+_I$} {\Gamma \vdash \ve{t}: T} {\Gamma \vdash \ve{r}: R}
    {\Gamma \vdash \ve{t} + \ve{r}: T + R}
  \]
  Trivial. \inductioncase{Case $\equiv$}
  \[
    \prftree[r]{$\equiv$} {\Gamma \vdash \ve{t} + \ve{r}: P} {S \equiv P}
    {\Gamma \vdash \ve{t} + \ve{r}: S}
  \]
  By the induction hypothesis, ${S \equiv P \equiv T + R}$. \inductioncase{Case
    $1_E$}
  \[
    \prftree[r]{$1_E$} {\pi = \Gamma \vdash 1\cdot(\ve{t} + \ve{r}): T} {\Gamma
      \vdash \ve{t} + \ve{r}: T}
  \]
  By Lemma~\ref{lem:sr:scalars}, there exist $R_1, \dots, R_m$, $\beta_1, \dots,
  \beta_m$ such that
  \begin{itemize}
  \item ${T \equiv \suj{m} \beta_j \cdot R_j}$.
  \item $\pi_j = \Gamma \vdash \ve{t} + \ve{r}: R_j$ with $size(\pi) >
    size(\pi_j)$ for $j \in \{1, \dots, m\}$.
  \item $\suj{m} \beta_j = 1$
  \end{itemize}
  Since $size(\pi) > size(\pi_j)$, by applying the induction hypothesis for all
  $j \in \{1, \dots, m\}$,
  \begin{itemize}
  \item $R_j \equiv S_{(j,1)} + S_{(j,2)}$.
  \item $\Gamma\vdash\ve t: S_{(j,1)}$.
  \item $\Gamma\vdash\ve r: S_{(j,2)}$.
  \end{itemize}
  Then,
  \[
    T \equiv \suj{m} \beta_j \cdot R_j \equiv \suj{m} \beta_j \cdot (S_{(j,1)} +
    S_{(j,2)}) \equiv \suj{m} \beta_j \cdot S_{(j,1)} + \suj{m} \beta_j \cdot
    S_{(j,2)}
  \]
  We can rewrite $T$ as follows:
  \[
    P_1 = \suj{m} \beta_j \cdot S_{(j,1)}\qquad P_2 = \suj{m} \beta_j \cdot
    S_{(j,2)}\qquad T \equiv P_1 + P_2
  \]
  Finally, we must prove that $\Gamma \vdash \ve{t}: P_1$ and $\Gamma \vdash \ve{r}: P_2$.\\
  Since $\Gamma\vdash\ve t: S_{(j,1)}$ and $\Gamma\vdash\ve r: S_{(j,2)}$ for
  all $j \in \{1,\dots,m\}$, applying the $S$ rule in both cases we
  have 
  \[
    \prftree[r]{$S$} {\Gamma\vdash\ve t: S_{(j,1)}~\forall j \in \{1,\dots,m\}}
    {\Gamma \vdash 1\cdot \ve{t}: P_1} \qquad \prftree[r]{$S$} {\Gamma\vdash\ve
      t: S_{(j,2)}~\forall j \in \{1,\dots,m\}}
    {\Gamma \vdash 1\cdot \ve{r}: P_2}
  \]
  Applying the $1_E$ rule to both sequents, we have
  \[
    \Gamma\vdash \ve t: P_1 \qquad \Gamma\vdash \ve r: P_2
  \]
  Finally, by $\equiv$ rule, $\Gamma \vdash \ve{t} + \ve{r}: T$.
  \inductioncase{Case $\forall$}
  \[
    \prftree[r]{$\forall$} {\Gamma \vdash \ve{t} + \ve{r}: \sui{n}\alpha_i\cdot
      U_i} {\Gamma \vdash \ve{t} + \ve{r}: \sui{n}\alpha_i\cdot V_i}
  \]
  Rules $\forall_I$ and $\forall_E$ both have the same structure as shown above.
  In any case, by the induction hypothesis $\Gamma\vdash\ve t:T$ and
  $\Gamma\vdash\ve r:R$ with
  $T+R\equiv\sui{n}\alpha_i\cdot U_i$.\\
  Then, there exist $N,M\subseteq\{1,\dots,n\}$ with $N\cup M=\{1,\dots,n\}$
  such that
  \begin{align*}
    T\equiv\sum_{i\in N\setminus M}\alpha_i\cdot U_i+\sum_{i\in N\cap M}\alpha_i'\cdot U_i&\qquad\textrm{and}\qquad
                                                                                            R\equiv\sum_{i\in M\setminus N}\alpha_i\cdot U_i+\sum_{i\in N\cap M}\alpha_i''\cdot U_i&
  \end{align*}
  where $\forall i\in N\cap M$, $\alpha_i'+\alpha_i''=\alpha_i$.\\
  Therefore, using $\equiv$ (if needed) and the same $\forall$-rule,
  \begin{align*}
    T\equiv\sum_{i\in N\setminus M}\alpha_i\cdot V_i+\sum_{i\in N\cap M}\alpha_i'\cdot V_i&\qquad\textrm{and}\qquad
                                                                                            R\equiv\sum_{i\in M\setminus N}\alpha_i\cdot V_i+\sum_{i\in N\cap M}\alpha_i''\cdot V_i&
  \end{align*}
\end{proof}

\xrecap{Lemma}{Application}{lem:sr:app}{ If $\Gamma\vdash(\ve t)~\ve r:T$, there
  exist $R_1, \dots, R_h$, $\mu_1, \dots, \mu_h$, $\V_1,\dots,\V_h$ such that $T
  \equiv \suk{h} \mu_k \cdot R_k$, $\suk{h} \mu_k = 1$ and for all $k \in
  \{1,\dots,h\}$
  \begin{itemize}
	\item $\Gamma\vdash\ve t: \sui{n_k}{\alpha_{(k,i)} \cdot\forall\vec{X}.(U\to
      T_{(k,i)})}$.
	\item $\Gamma\vdash\ve r: \suj{m_k}\beta_{(k,j)}\cdot
    U[\vec{A}_{(k,j)}/\vec{X}]$.
  \item $\sui{n_k}\suj{m_k} \alpha_{(k,i)}\times\beta_{(k,j)}\cdot
    {T_{(k,i)}[\vec{A}_{(k,j)}/\vec{X}]} \ssubt_{\V_k,\Gamma} R_k$.
  \end{itemize}
}
\begin{proof}
  By induction on the typing derivation.
  \inductioncase{Case $\to_E$}
  \[
    \prftree[r]{$\to_E$}
    {\Gamma \vdash \ve{t}: \sui{n}{\alpha_i \cdot\forall\vec{X}.(U\to T_i)}}
    {\Gamma \vdash \ve{r}: \suj{m}\beta_j\cdot U[\vec{A}_j/\vec{X}]}
    {\Gamma \vdash (\ve t)~\ve r: \sui{n}\suj{m} \alpha_i\times\beta_j\cdot {T_i[\vec{A}_j/\vec{X}]}}
  \]
  Take $\mu_1, \dots, \mu_h$ such that $\suk{h} \mu_k = 1$, then
  \[
    \sui{n}\suj{m} \alpha_i\times\beta_j\cdot T_i[\vec{A}_j/\vec{X}] \equiv \suk{h} \mu_k \cdot \sui{n}\suj{m} \alpha_i\times\beta_j\cdot T_i[\vec{A}_j/\vec{X}]
  \]
  So this is the trivial case.
  \inductioncase{Case $\equiv$}
  \[
    \prftree[r]{$\equiv$}
    {\Gamma \vdash (\ve t)~\ve r: P}
    {S \equiv P}
    {\Gamma \vdash (\ve t)~\ve r: S}
  \]
  By the induction hypothesis, there exist $R_1, \dots, R_h$, $\mu_1, \dots, \mu_h$, $\V_1,\dots,\V_h$ such that $P \equiv S \equiv \suk{h} \mu_k \cdot R_k$, $\suk{h} \mu_k = 1$
  and for all $k \in \{1,\dots,h\}$,
  \begin{itemize}
  \item $\Gamma\vdash\ve t: \sui{n_k}{\alpha_{(k,i)} \cdot\forall\vec{X}.(U\to T_{(k,i)})}$.
  \item $\Gamma\vdash\ve r: \suj{m_k}\beta_{(k,j)}\cdot U[\vec{A}_{(k,j)}/\vec{X}]$.
  \item $\sui{n_k}\suj{m_k} \alpha_{(k,i)}\times\beta_{(k,j)}\cdot {T_{(k,i)}[\vec{A}_{(k,j)}/\vec{X}]} \ssubt_{\V_k,\Gamma} R_k$.
  \end{itemize}
  So the lemma holds.
  \inductioncase{Case $1_E$}
  \[
    \prftree[r]{$1_E$}
    {\pi = \Gamma \vdash 1\cdot(\ve t)~\ve r: T}
    {\Gamma \vdash (\ve t)~\ve r: T}
  \]
  By Lemma~\ref{lem:sr:scalars}, there exist $R_1, \dots, R_h$, $\mu_1, \dots, \mu_h$ such that
  \begin{itemize}
  \item $T \equiv \suk{h} \mu_k \cdot R_k$.
  \item $\pi_k = \Gamma \vdash (\ve t)~\ve r: R_k$, with $size(\pi) > size(\pi_k)$, for $k \in \{1, \dots, h\}$..
  \item $\suk{h} \mu_k = 1$.
  \end{itemize}
  Since $size(\pi) > size(\pi_k)$, we apply the inductive hypothesis for all $k \in \{1, \dots, h\}$ (and omiting the $k$ index for readability),
  so there exist $S_{1}, \dots, S_{p}$, $\eta_{1}, \dots, \eta_{p}$, $\V_{1},\dots,\V_{p}$  such that
  $R \equiv \sug{q}{p} \eta_{q} \cdot S_{q}$, $\sug{q}{p} \eta_{q} = 1$ and for all $q \in \{1,\dots,p\}$,
  \begin{itemize}
  \item $\Gamma\vdash\ve t: \sui{n_{q}}{\alpha_{(q,i)} \cdot\forall\vec{X}.(U\to T_{(q,i)})}$.
  \item $\Gamma\vdash\ve r: \suj{m_{q}}\beta_{(q,j)}\cdot U[\vec{A}_{(q,j)}/\vec{X}]$.
  \item $\sui{n_{q}}\suj{m_{q}} \alpha_{(q,i)}\times\beta_{(q,j)}\cdot {T_{(q,i)}[\vec{A}_{(q,j)}/\vec{X}]} \ssubt_{\V_{q},\Gamma} S_{q}$.
  \end{itemize}
  Then
  \[
    T \equiv \suk{h} \mu_k \cdot R_k \equiv \suk{h} \mu_k \cdot \sug{q}{p_k} \eta_{(k,q)} \cdot S_{(k,q)} \equiv \suk{h}\sug{q}{p_k} (\mu_k \times  \eta_{(k,q)}) \cdot S_{(k,q)}
  \]
  Finally, we must prove that $\suk{h}\sug{q}{p_k} (\mu_k \times  \eta_{(k,q)}) = 1$,
  \[
    \suk{h}\sug{q}{p_k} (\mu_k \times  \eta_{(k,q)}) = \suk{h} \mu_k \cdot \underbrace{\sug{q}{p_k} \eta_{(k,q)}}_{=~1} = \suk{h} \mu_k = 1
  \]

  \inductioncase{Case $\forall_I$}
  \[
    \prftree[r]{$\forall_I$}
    {\pi' = \Gamma \vdash (\ve{t})~\ve{r}: \sug{a}{b}\sigma_a \cdot V_a}
    {X \notin \FV{\Gamma}}
    {\Gamma \vdash (\ve{t})~\ve{r}: \sug{a}{b}\sigma_a \cdot \forall X. V_a}
  \]
  By the induction hypothesis there exist $R_1,\dots,R_h$, $\mu_1,\dots,\mu_h$, $\V_1,\dots,\V_h$ such that
  $\sug{a}{b}\sigma_{a} \cdot V_a \equiv \suk{h} \mu_k \cdot R_k$, $\suk{h} \mu_k = 1$ and for all $k \in \{1,\dots,h\}$,
  \begin{itemize}
  \item $\Gamma\vdash\ve t: \sui{n_k}{\alpha_{(k,i)} \cdot\forall\vec{X}.(U\to T_{(k,i)})}$.
  \item $\Gamma\vdash\ve r: \suj{m_k}\beta_{(k,j)}\cdot U[\vec{A}_{(k,j)}/\vec{X}]$.
  \item $\sui{n_k}\suj{m_k} \alpha_{(k,i)}\times\beta_{(k,j)}\cdot {T_{(k,i)}[\vec{A}_{(k,j)}/\vec{X}]} \ssubt_{\V_k,\Gamma} R_k$.
  \end{itemize}
  By Lemma~\ref{lem:sr:typecharact}, and since $\sug{a}{b}\sigma_{a} \cdot V_a$ does not have any general variable,
  then for all $k \in \{1,\dots,h\}$, $R_k \equiv \sug{c}{d_k}\eta_{(k,c)}\cdot W_{(k,c)}$.\\
  Hence $\sug{a}{b}\sigma_{a} \cdot V_a \equiv \suk{h}\mu_h \cdot \sug{c}{d_k}\eta_{(k,c)}\cdot W_{(k,c)}$.\\
  Without loss of generality, assuming all unit types are distinct (not equivalent),
  then by Lemma~\ref{lem:sr:equivforall},
  \[
    \sug{a}{b}\sigma_a \cdot \forall X. V_a \equiv
    \suk{h}\mu_k \cdot \underbrace{\sug{c}{d_k}\eta_{(k,c)}\cdot \forall X.W_{(k,c)}}_{R'_k}
  \]
  Finally, for all $k \in \{1,\dots,h\}$ we must prove that
  $\sui{n_k}\suj{m_k} \alpha_{(k,i)}\times\beta_{(k,j)}\cdot {T_{(k,i)}[\vec{A}_{(k,j)}/\vec{X}]} \ssubt_{\V'_k,\Gamma} R'_k$.\\
  Notice that $R_k \ssubt_{\V_k \cup \{X\},\Gamma} R'_k$, then by definition of $\ssubt$, taking $\V'_k = \V_k \cup \{X\}$,\\
  $\sui{n_k}\suj{m_k} \alpha_{(k,i)}\times\beta_{(k,j)}\cdot {T_{(k,i)}[\vec{A}_{(k,j)}/\vec{X}]} \ssubt_{\V'_k,\Gamma} R'_k$.

  \inductioncase{Case $\forall_E$}
  \[
    \prftree[r]{$\forall_E$}
    {\Gamma \vdash (\ve{t})~\ve{r}: \sug{a}{b}\sigma_a \cdot \forall X.V_a}
    {\Gamma \vdash (\ve{t})~\ve{r}: \sug{a}{b-1}\sigma_a \cdot V_a[A/X]}
  \]
  By the induction hypothesis there exist $R_1,\dots,R_h$, $\mu_1,\dots,\mu_h$, $\V_1,\dots,\V_h$ such that
  $\sug{a}{b}\sigma_a \cdot \forall X.V_a \equiv \suk{h} \mu_k \cdot R_k$,
  $\suk{h} \mu_k = 1$ and for all $k \in \{1,\dots,h\}$,
  \begin{itemize}
  \item $\Gamma\vdash\ve t: \sui{n_k}{\alpha_{(k,i)} \cdot\forall\vec{X}.(U\to T_{(k,i)})}$.
  \item $\Gamma\vdash\ve r: \suj{m_k}\beta_{(k,j)}\cdot U[\vec{A}_{(k,j)}/\vec{X}]$.
  \item $\sui{n_k}\suj{m_k} \alpha_{(k,i)}\times\beta_{(k,j)}\cdot {T_{(k,i)}[\vec{A}_{(k,j)}/\vec{X}]} \ssubt_{\V_k,\Gamma} R_k$.
  \end{itemize}
  By Lemma~\ref{lem:sr:typecharact}, and since $\sug{a}{b}\sigma_a \cdot \forall X.V_a$ does not have any general variable,
  $R_k \equiv \sug{c}{d_k}\eta_{(k,c)}\cdot W_{(k,c)}$.\\
  Hence $\sug{a}{b}\sigma_a \cdot \forall X.V_a \equiv \suk{h} \mu_k \cdot \sug{c}{d_k}\eta_{(k,c)}\cdot W_{(k,c)}$.\\
  Without loss of generality, we assume that all unit types present at both sides of the equivalence are distinct,
  then by Lemma~\ref{lem:sr:equivforall}, for all $k \in \{1,\dots,h\}, c \in \{1,\dots,d_k\}$, there exists $W'_{(k,c)}$
  such that $W_{(k,c)} \equiv \forall X.W'_{(k,c)}$, so we have
  \[
    \sug{a}{b}\sigma_a \cdot \forall X. V_a \equiv
    \suk{h}\mu_k \cdot \underbrace{\sug{c}{d_k}\eta_{(k,c)}\cdot \forall X.W'_{(k,c)}}_{R_k}
  \]
  By the same lemma, we have that
  \[
    \sug{a}{b}\sigma_a \cdot V_a[A/X] \equiv
    \suk{h}\mu_k \cdot \underbrace{\sug{c}{d_k}\eta_{(k,c)}\cdot W'_{(k,c)}[A/X]}_{R'_k}
  \]
  Finally, for all $k \in \{1,\dots,h\}$ we must prove that
  $\sui{n_k}\suj{m_k} \alpha_{(k,i)}\times\beta_{(k,j)}\cdot {T_{(k,i)}[\vec{A}_{(k,j)}/\vec{X}]} \ssubt_{\V'_k,\Gamma} R'_k$.\\
  Notice that $R_k \ssubt_{\V_k \cup \{X\},\Gamma} R'_k$, then by definition of $\ssubt$, taking $\V'_k = \V_k \cup \{X\}$,\\
  $\sui{n_k}\suj{m_k} \alpha_{(k,i)}\times\beta_{(k,j)}\cdot {T_{(k,i)}[\vec{A}_{(k,j)}/\vec{X}]} \ssubt_{\V'_k,\Gamma} R'_k$.
\end{proof}

\xrecap{Lemma}{Abstractions}{lem:sr:abs}{
  If $\Gamma\vdash\lambda x.\ve t:T$, then there exist $T_1,\dots,T_n$, $R_1,\dots,R_n$, $U_1,\dots,U_n$,
  $\alpha_1,\dots,\alpha_n$, $\V_1,\dots,\V_n$ such
  that $T \equiv \sui{n} \alpha_i \cdot T_i$, $\sui{n} \alpha_i = 1$ and for all $i \in \{1,\dots,n\}$,
  \begin{itemize}
  \item $\Gamma,x:U_i\vdash\ve t:R_i$.
  \item $U_i \to R_i \ssubt_{\V_i,\Gamma} T_i$.
  \end{itemize}
}
\begin{proof}
  By induction on the typing derivation
  \inductioncase{Case $\to_I$}
  \[
    \prftree[r]{$\to_I$}
    {\Gamma,x:U\vdash\ve t:R}
    {\Gamma\vdash\lambda x.\ve t: U \to R}
  \]
  Trivial.
  \inductioncase{Case $\equiv$}
  \[
    \prftree[r]{$\equiv$}
    {\Gamma\vdash\lambda x.\ve t: R}
    {R \equiv T}
    {\Gamma \vdash \lambda x.\ve t: T}
  \]
  By the induction hypothesis, there exist $T_1,\dots,T_n$, $R_1,\dots,R_n$, $U_1,\dots,U_n$,
  $\alpha_1,\dots,\alpha_n$, $\V_1,\dots,\V_n$ such
  that $T \equiv R \equiv \sui{n} \alpha_i \cdot T_i$, $\sui{n} \alpha_i = 1$ and for all $i \in \{1,\dots,n\}$,
  \begin{itemize}
  \item $\Gamma,x:U_i\vdash\ve t:R_i$.
  \item $U_i \to R_i \ssubt_{\V_i,\Gamma} T_i$.
  \end{itemize}
  So the lemma holds.
  \inductioncase{Case $1_E$}
  \[
    \prftree[r]{$1_E$}
    {\pi = \Gamma\vdash 1\cdot(\lambda x.\ve t): T}
    {\Gamma\vdash \lambda x.\ve t: T}
  \]
  By Lemma~\ref{lem:sr:scalars}, there exist $R_1, \dots, R_m$, $\beta_1, \dots, \beta_m$ such that
  \begin{itemize}
  \item $T \equiv \suj{m}\beta_i \cdot R_j$.
  \item $\pi_i = \Gamma \vdash \ve{t}: R_j$, with $size(\pi) > size(\pi_j)$, for $j \in \{1, \dots, n\}$.
  \item $\suj{n} \beta_i = 1$.
  \end{itemize}
  Since $size(\pi) > size(\pi_j)$, by induction hypothesis, for all $j \in \{1, \dots, n\}$
  there exist $S_{(j,1)},\dots,S_{(j,n_j)}$, $P_{(j,1)},\dots,P_{(j,n_j)}$, $U_{(j,1)},\dots,U_{(j,n_j)}$,
  $\eta_{(j,1)},\dots,\eta_{(j,n_j)}$, $\V_{(j,1)},\dots,\V_{(j,n_j)}$ such
  that $R_j \equiv \sui{n_j} \eta_{(j,i)} \cdot S_{(j,i)}$, $\sui{n_j} \eta_{(j,i)} = 1$ and for all $i \in \{1,\dots,n_j\}$,
  \begin{itemize}
  \item $\Gamma,x:U_{(j,i)}\vdash\ve t:P_{(j,i)}$.
  \item $U_{(j,i)} \to P_{(j,i)} \ssubt_{\V_{(j,i)},\Gamma} S_{(j,i)}$.
  \end{itemize}
  Then we have
  \[
    T \equiv \suj{m} \beta_j \cdot R_j \equiv \suj{m} \beta_j \cdot \sui{n_j} \eta_{(j,i)} \cdot S_{(j,i)} \equiv \suj{m}\sui{n_j} (\beta_j \times  \eta_{(j,i)}) \cdot S_{(j,i)}
  \]
  Finally, we must prove that $\suj{m}\sui{n_j} (\beta_j \times  \eta_{(j,i)}) = 1$:
  \[
    \suj{m}\sui{n_j} (\beta_j \times  \eta_{(j,i)}) = \suj{m} \beta_j \cdot \underbrace{\sui{n_j} \eta_{(j,i)}}_{=~1} = \suj{m} \beta_j = 1
  \]
  \inductioncase{Case $\forall_I$}
  \[
    \prftree[r]{$\forall_I$}
    {\Gamma \vdash \lambda x.\ve t: \sui{n} \alpha_i \cdot U_i}
    {\prfassumption{X \notin \FV{\Gamma}}}
    {\Gamma \vdash \lambda x.\ve t: \sui{n} \alpha_i \cdot \forall X.U_i}
  \]
  By the induction hypothesis, there exist $T_1,\dots,T_m$, $R_1,\dots,R_m$, $V_1,\dots,V_m$,
  $\alpha_1,\dots,\alpha_m$, $\V_1,\dots,\V_m$ such
  that $\sui{n} \alpha_i \cdot U_i \equiv \suj{m} \mu_j \cdot T_j$, $\suj{m} \mu_i = 1$ and for all $j \in \{1,\dots,m\}$,
  \begin{itemize}
  \item $\Gamma,x:V_j\vdash\ve t:R_j$.
  \item $V_j \to R_j \ssubt_{\V_j,\Gamma} T_j$.
  \end{itemize}
  By Lemma~\ref{lem:sr:typecharact}, and since $\sui{n} \alpha_i \cdot U_i$ does not have any general variable $\vara{X}$, then
  $T_i \equiv \sug{k}{h_j} \beta_{(j,k)} \cdot W_{(j,k)}$.
  Hence $\sui{n} \alpha_i \cdot U_i \equiv \suj{m} \mu_i \cdot \sug{k}{h_j} \beta_{(j,k)} \cdot W_{(j,k)}$.
  Without loss of generality, assuming all unit types are distinct (not equivalent),
  then by Lemma~\ref{lem:sr:equivforall},
  \[
    \sui{n} \alpha_i \cdot U_i \equiv
    \suj{m} \mu_i \cdot \underbrace{\sug{k}{h_j} \beta_{(j,k)} \cdot \forall X.W_{(j,k)}}_{T'_j}
  \]
  Finally, we must prove that $V_j \to R_j \ssubt_{\V'_j,\Gamma} T'_j$ for some $\V'_j$.
  Since $V_j \to R_j \ssubt_{\V_j,\Gamma} T_j$ and $T_j \ssubt_{\V''_j,\Gamma} T'_j$, then by $\ssubt$ and using $\V'_j = \V_j \cup \V''_j$,
  we conclude that $V_j \to R_j \ssubt_{\V'_j,\Gamma} T'_j$.
  \inductioncase{Case $\forall_E$}
  \[
    \prftree[r]{$\forall_E$}
    {\Gamma \vdash \lambda x.\ve t: \sui{n} \alpha_i \cdot \forall X.U_i}
    {\Gamma \vdash \lambda x.\ve t: \sui{n} \alpha_i \cdot U_i[A/X]}
  \]
  By the induction hypothesis, there exist $T_1,\dots,T_m$, $R_1,\dots,R_m$, $V_1,\dots,V_m$,
  $\alpha_1,\dots,\alpha_m$, $\V_1,\dots,\V_m$ such
  that $\sui{n} \alpha_i \cdot \forall X.U_i \equiv \suj{m} \mu_j \cdot T_j$, $\suj{m} \mu_i = 1$ and for all $j \in \{1,\dots,m\}$,
  \begin{itemize}
  \item $\Gamma,x:V_j\vdash\ve t:R_j$.
  \item $V_j \to R_j \ssubt_{\V_j,\Gamma} T_j$.
  \end{itemize}
  By Lemma~\ref{lem:sr:typecharact}, and since $\sui{n} \alpha_i \cdot U_i$ does not have any general variable $\vara{X}$, then
  $T_i \equiv \sug{k}{h_j} \beta_{(j,k)} \cdot W_{(j,k)}$.
  Hence $\sui{n} \alpha_i \cdot U_i \equiv \suj{m} \mu_i \cdot \sug{k}{h_j} \beta_{(j,k)} \cdot W_{(j,k)}$.
  Without loss of generality, assuming all unit types are distinct (not equivalent),
  then by Lemma~\ref{lem:sr:equivforall}, for all $j \in \{1, \dots, m\}$, $k \in \{1,\dots,h_j\}$,
  there exists $W'_{(j,k)}$ such that $W_{(j,k)} \equiv \forall X.W'_{(j,k)}$.
  Then,
  \[
    \sui{n} \alpha_i \cdot \forall X.U_i \equiv
    \suj{m} \mu_i \cdot \underbrace{\sug{k}{h_j} \beta_{(j,k)} \cdot \forall X.W'_{(j,k)}}_{T_j}
  \]
  By the same lemma, we have that
  \[
    \sui{n} \alpha_i \cdot U_i[A/X] \equiv
    \suj{m} \mu_i \cdot \underbrace{\sug{k}{h_j} \beta_{(j,k)} \cdot W'_{(j,k)}[A/X]}_{T'_j}
  \]
  Finally, we must prove that $V_j \to R_j \ssubt_{\V'_j,\Gamma} T'_j$ for some $\V'_j$.
  Since $V_j \to R_j \ssubt_{\V_j,\Gamma} T_j$ and $T_j \ssubt_{\V''_j,\Gamma} T'_j$, then by $\ssubt$ and using $\V'_j = \V_j \cup \V''_j$,
  we conclude that $V_j \to R_j \ssubt_{\V'_j,\Gamma} T'_j$.
\end{proof}

\xrecap{Lemma}{Basis terms}{lem:sr:basevectors}{
  For any context $\Gamma$, type $T$ and basis term $\ve{b}$, if
  $\Gamma\vdash\ve{b}: T$ there exist $U_1, \dots, U_n$, $\alpha_1, \dots, \alpha_n$ such that
  \begin{itemize}
  \item $T \equiv \sui{n} \alpha_i \cdot U_i$.
  \item $\Gamma\vdash\ve{b}: U_i$, for $i \in \{1,\dots,n\}$.
  \item $\sui{n} \alpha_i = 1$.
  \end{itemize}
}
\begin{proof}
  By induction on the typing derivation.
  \inductioncase{Case $ax$}
  \[
    \prftree[r]{$ax$}
    {}
    {\Gamma, x:{U}\vdash x:{U}}
    \raisebox{9pt}{\qquad\text{and}\qquad}
    \prftree[r]{$\to_I$}
    {\Gamma,x:U\vdash\ve t:T}
    {\Gamma\vdash\lambda x.\ve t: U \to T}
  \]
  Trivial cases.
  \inductioncase{Case $\equiv$}
  \[
    \prftree[r]{$\equiv$}
    {\Gamma\vdash\ve{b}: R}
    {R \equiv T}
    {\Gamma \vdash \ve{b}: T}
  \]
  By the induction hypothesis, there exist $U_1, \dots, U_n$, $\alpha_1, \dots, \alpha_n$ such that
  \begin{itemize}
  \item $T \equiv R \equiv \sui{n} \alpha_i \cdot U_i$.
  \item $\Gamma\vdash\ve{b}: U_i$, for $i \in \{1,\dots,n\}$.
  \item $\sui{n} \alpha_i = 1$.
  \end{itemize}
  So the lemma holds.
  \inductioncase{Case $1_E$}
  \[
    \prftree[r]{$1_E$}
    {\pi = \Gamma\vdash 1\cdot\ve{b}: T}
    {\Gamma\vdash \ve{b}: T}
  \]
  By Lemma~\ref{lem:sr:scalars}, there exist $R_1, \dots, R_m$, $\beta_1, \dots, \beta_m$ such that
  \begin{itemize}
  \item $T \equiv \suj{m}\beta_j \cdot R_j$.
  \item $\suj{m} \beta_j = 1$, and $\pi_j = \Gamma \vdash \ve{b}: R_j$ with $size(\pi) > size(\pi_j)$ for $j = \{1, \dots, m\}$.
  \item $\suj{m} \beta_j = 1$.
  \end{itemize}
  Since $size(\pi) > size(\pi_j)$, by induction hypothesis, for all $j = \{1, \dots, m\}$ there exist
  $U_{(j,1)},\dots,U_{(j,n_j)}$, $\alpha_{(j,1)},\dots,\alpha_{(j,n_j)}$ such that
  \begin{itemize}
  \item $R_j \equiv \sui{n_j} \alpha_{(j,i)} \cdot U_{(j,i)}$.
  \item $\Gamma\vdash\ve{b}: U_{(j,i)}$, for $i \in \{1,\dots,n_j\}$.
  \item $\sui{n_j} \alpha_{(j,i)} = 1$.
  \end{itemize}
  Then
  \[
    T \equiv \suj{m} \beta_j \cdot R_j \equiv \suj{m} \beta_j \cdot \sui{n_j} \alpha_{(j,i)} \cdot U_{(j,i)} \equiv \suj{m}\sui{n_j} (\beta_j \times \alpha_{(j,i)}) \cdot U_{(j,i)}
  \]
  Finally, we must prove that $\suj{m}\sui{n_j} (\beta_j \times \alpha_{(j,i)}) = 1$:
  \[
    \suj{m}\sui{n_j} (\beta_j \times \alpha_{(j,i)}) = \suj{m} \beta_j \cdot \underbrace{\sui{n_j} \alpha_{(j,i)}}_{=~1} = \suj{m} \beta_j = 1
  \]
  \inductioncase{Case $\forall$}
  \[
    \prftree[r]{$\forall$}
    {\Gamma \vdash \ve{b}: \sui{n} \alpha_i \cdot U_i}
    {\Gamma \vdash \ve{b}: \sui{n} \alpha_i \cdot V_i}
  \]
  $\forall$-rules ($\forall_I$ and $\forall_E$) both have the same structure as shown above.\\
  In both cases, by the induction hypothesis, there exist $W_1, \dots, W_m$, $\beta_1, \dots, \beta_m$ such that
  \begin{itemize}
  \item $\sui{n} \alpha_i \cdot U_i \equiv \suj{m} \beta_j \cdot W_j$.
  \item $\Gamma\vdash\ve{b}: W_j$, for $j \in \{1,\dots,m\}$.
  \item $\suj{m} \beta_j = 1$.
  \end{itemize}
  Without loss of generality, we assume that all unit types present at both sides of the equivalence are distinct,
  so by Lemma~\ref{lem:sr:equivdistinctscalars},
  then $m = n$ and there exists a permutation $p$ of $m$ such that
  for all $i \in \{1,\dots,n\}$, then $U_i = W_{p(i)}$ and $\alpha_i = \beta_{p(i)}$,
  which means that $\sui{n} \alpha_i = 1$.
  Finally, by applying the corresponding $\forall$ rule for all $i \in \{1,\dots,n\}$, we have
  \[
    \prftree[r]{$\forall$}
    {\Gamma \vdash \ve{b}: U_i}
    {\Gamma \vdash \ve{b}: V_i}
  \]
\end{proof}

\xrecap{Lemma}{Substitution lemma}{lem:sr:substitution}{
  For any term ${\ve t}$, basis term $\ve b$, term variable $x$, context $\Gamma$, types $T$, $U$, type variable $X$ and type $A$, where $A$ is a unit type if $X$ is a unit variable, otherwise $A$ is a general type, we have,
  \begin{enumerate}
  \item If $\Gamma\vdash\ve{t}: T$, then $\Gamma[A/X]\vdash\ve{t}: T[A/X]$;
  \item If $\Gamma,x:U\vdash\ve t:T$ and $\Gamma\vdash\ve b:U$, then $\Gamma\vdash\ve t[\ve b/x]: T$.
  \end{enumerate}
}
\begin{proof}~
  \textleadbydots{Item (1)}
  Induction on the typing derivation.
  \inductioncase{Case $ax$}
  \[
    \prftree[r]{$ax$}
    {}
    {\Gamma, x: {U}: \vdash x: U}
  \]
  Notice that ${\Gamma[A/X],x:U[A/X]\vdash x:U[A/X]}$ can also be derived with the same rule.
  \inductioncase{Case $\to_I$}
  \[
    \prftree[r]{$\to_I$}
    {\Gamma,x:U\vdash\ve t:T}
    {\Gamma\vdash\lambda x.\ve t:U\to T}
  \]
  By the induction hypothesis $\Gamma[A/X],x:U[A/X]\vdash\ve t:T[A/X]$, so by rule $\to_I$, $\Gamma[A/X]\vdash\lambda x.\ve t:U[A/X]\to T[A/X]=(U\to T)[A/X]$.
  \inductioncase{Case $\to_E$}
  \[
    \prftree[r]{$\to_E$}
    {\Gamma\vdash\ve t:\sui{n}\alpha_i\cdot\forall\vec Y.(U\to T_i)}
    {\Gamma\vdash\ve r:\suj{m}\beta_j\cdot U[\vec B_j/\vec Y]}
    {\Gamma\vdash(\ve t)~\ve r:\sui{n}\suj{m}\alpha_i\times\beta_j\cdot T_i[\vec B_j/\vec Y]}
  \]
  By the induction hypothesis
  $\Gamma[A/X]\vdash\ve t:(\sui{n}\alpha_i\cdot\forall\vec Y.(U\to T_i))[A/X]$ and this type is equal to $\sui{n}\alpha_i\cdot\forall\vec Y.(U[A/X]\to T_i[A/X])$.
  Also $\Gamma[A/X]\vdash\ve r:(\suj{m}\beta_j\cdot U[\vec B_j/\vec Y])[A/X]=
  \suj{m}\beta_j\cdot U[\vec B_j/\vec Y][A/X]$.
  Since $\vec Y$ is bound, we can consider $\vec{Y} \notin \FV{A}$.
  Hence $U[\vec B_j/\vec Y][A/X]=U[A/X][\vec B_j[A/X]/\vec Y]$, and so, by rule $\to_E$,
  \begin{align*}
    \Gamma[A/X]\vdash(\ve t)~\ve r&:\sui{n}\suj{m}\alpha_i\times\beta_j\cdot T_i[A/X][\vec B_j[A/X]/\vec Y]\\
                                  &=\left(\sui{n}\suj{m}\alpha_i\times\beta_j\cdot T_i[\vec B_j/\vec Y]\right)[A/X]
  \end{align*}
  \inductioncase{Case $\forall_I$}
  \[
    \prftree[r]{$\forall_I$}
    {\Gamma\vdash\ve t:\sui{n}\alpha_i\cdot U_i}
    {Y\notin\FV{\Gamma}}
    {\Gamma\vdash\ve{t}:\sui{n}\alpha_i\cdot \forall Y.U_i }
  \]
  By the induction hypothesis,
  $\Gamma[A/X]\vdash\ve t:(\sui{n}\alpha_i\cdot
  U_i)[A/X]=\sui{n}\alpha_i\cdot U_i[A/X]$.
  Then, by
  rule $\forall_I$, $\Gamma[A/X]\vdash\ve
  t:\sui{n}\alpha_i\cdot \forall Y.U_i[A/X]=(\sui{n}\alpha_i\cdot \forall Y.U_i)[A/X]$.
  Since $Y$ is bound, we can consider $Y \notin \FV{A}$.
  \inductioncase{Case $\forall_E$}
  \[
    \prftree[r]{$\forall_E$}
    {\Gamma\vdash\ve t:\sui{n}\alpha_i\cdot \forall Y.U_i}
    {\Gamma\vdash\ve t:\sui{n}\alpha_i\cdot U_i[B/Y]}
  \]
  By the induction
  hypothesis $\Gamma[A/X]\vdash\ve
  t:(\sui{n}\alpha_i\cdot \forall Y.U_i)[A/X]=\sui{n}\alpha_i\cdot \forall Y.U_i[A/X]$.
  Since $Y$ is bound, we can
  consider $Y \notin \FV{A}$.
  Then by rule $\forall_E$,
  ${\Gamma[A/X]\vdash\ve t:\sui{n}\alpha_i\cdot U_i[A/X][B/Y]}$.
  We can consider $X\notin\FV{B}$ (in
  other case, just take $B[A/X]$ in the $\forall$-elimination), hence
  \[
    \sui{n}\alpha_i\cdot U_i[A/X][B/Y] =\sui{n}\alpha_i\cdot U_i[B/Y][A/X] =\left(\sui{n}\alpha_i\cdot U_i[B/Y]\right)[A/X]
  \]
  \inductioncase{Case $S$}
  \[
    \prftree[r]{$S$}
    {\Gamma\vdash\ve t:T_i~\forall i \in \{1,\dots,n\}}
    {\Gamma\vdash \left(\sui{n} \alpha_i\right) \cdot\ve t: \sui{n} \alpha_i \cdot T_i}
  \]
  By the induction hypothesis, for all $i \in \{1, \dots, n\}$,  $\Gamma[A/X]\vdash\ve t:T_i[A/X]$,
  so by rule $S$,
  $\Gamma[A/X]\vdash \left(\sui{n} \alpha_i\right) \cdot\ve t:\sui{n} \alpha_i \cdot T_i[A/X]={(\sui{n} \alpha_i \cdot T_i)[A/X]}$.
  \inductioncase{Case $+_I$}
  \[
    \prftree[r]{$+_I$}
    {\Gamma\vdash\ve t:T}
    {\Gamma\vdash\ve r:R}
    {\Gamma\vdash\ve t+\ve r:T+R}
  \]
  By the induction hypothesis $\Gamma[A/X]\vdash\ve t:T[A/X]$ and $\Gamma[A/X]\vdash\ve r:R[A/X]$, so by rule $+_I$, ${\Gamma[A/X]\vdash\ve t+\ve r:T[A/X]+R[A/X]=(T+R)[A/X]}$.
  \inductioncase{Case $\equiv$}
  \[
    \prftree[r]{$\equiv$}
    {\Gamma\vdash\ve t:T}
    {T\equiv R}
    {\Gamma\vdash\ve t:R}
  \]
  By the induction hypothesis $\Gamma[A/X]\vdash\ve t:T[A/X]$, and since $T\equiv R$, then $T[A/X]\equiv R[A/X]$, so by rule $\equiv$, $\Gamma[A/X]\vdash\ve t:R[A/X]$.
  \inductioncase{Case $1_E$}
  \[
    \prftree[r]{$1_E$}
    {\Gamma\vdash1\cdot\ve t:T}
    {\Gamma\vdash\ve t: T}
  \]
  By the induction hypothesis $\Gamma[A/X]\vdash1\cdot\ve t:T[A/X]$.
  By rule $1_E$, $\Gamma[A/X]\vdash\ve t:T[A/X]$.

  \medskip
  \textleadbydots{Item (2)}
  We proceed by induction on the typing derivation of $\Gamma,x:U\vdash\ve t:T$.
  \inductioncase{Case $ax$}
  \[
    \prftree[r]{$ax$}
    {\Gamma,x:U\vdash\ve t:T}
  \]
  Cases:
  \begin{itemize}
  \item $\ve t=x$, then $T=U$, and so $\Gamma\vdash\ve t[\ve b/x]:T$
    and $\Gamma\vdash\ve b:U$ are the same sequent.
  \item $\ve t=y$.
    Notice that $y[\ve b/x]=y$.
   \arxiv{By Lemma~\ref{lem:sr:weakening}}\conf{By weakening} $\Gamma,x:U\vdash y:T$ implies
    $\Gamma\vdash y:T$.
  \end{itemize}
  \inductioncase{Case $\to_I$}
  \[
    \prftree[r]{$\to_I$}
    {\Gamma,x:U,y:V\vdash \ve{r}:R}
    {\Gamma,x:U\vdash\lambda x.\lambda y.\ve{r}:V\to R}
  \]
  Since our system admits weakening \arxiv{(Lemma~\ref{lem:sr:weakening})}, the sequent $\Gamma,y:V\vdash\ve b:U$
  is derivable.
  Then by the induction hypothesis,
  $\Gamma,y:V\vdash\ve r[\ve b/x]:R$, from where, by
  rule $\to_I$, we obtain $\Gamma\vdash\lambda y.\ve r[\ve
  b/x]:V\to R$.
  We conclude, since
  $\lambda y.\ve r[\ve b/x]=(\lambda y.\ve r)[\ve b/x]$.
  \inductioncase{Case $\to_E$}
  \[
    \prftree[r]{$\to_E$}
    {\Gamma,x:U\vdash\ve r:\sui{n}\alpha_i\cdot\forall\vec Y.(V\to T_i)}
    {\Gamma,x:U\vdash\ve u:\suj{m}\beta_j\cdot V[\vec B/\vec Y]}
    {\Gamma,x:U\vdash(\ve r)~\ve u:\sui{n}\suj{m}\alpha_i\times\beta_j\cdot R_i[\vec B/\vec Y]}
  \]
  By the induction hypothesis,
  $\Gamma\vdash\ve r[\ve b/x]:\sui{n}\alpha_i\cdot\forall\vec Y.(V\to R_i)$
  and
  $\Gamma\vdash\ve u[\ve b/x]:\suj{m}\beta_j\cdot V[\vec B/\vec Y]$.
  Then, by rule $\to_E$,
  $\Gamma
  \vdash \ve r[\ve b/x]~\ve u[\ve b/x]:
  \sui{n}\suj{m}\alpha_i\times\beta_j\cdot R_i[\vec B/\vec Y]$.
  \inductioncase{Case $\forall_I$}
  \[
    \prftree[r]{$\forall_I$}
    {\Gamma,x:U\vdash\ve t:\sui{n}\alpha_i\cdot V_i}
    {Y\notin\FV{\Gamma}\cup\FV{U}}
    {\Gamma,x:U\vdash\ve{t}:\sui{n-1}\alpha_i\cdot V_i + \alpha_n\cdot\forall Y.V_n}
  \]
  By the induction hypothesis,
  $\Gamma\vdash\ve t[\ve b/x]:\sui{n}\alpha_i\cdot V_i$.
  Then by rule
  $\forall_I$, $\Gamma\vdash\ve t[\ve
  b/x]:\sui{n-1}\alpha_i\cdot V_i + \alpha_n\cdot\forall Y.V_n$.
  \inductioncase{Case $\forall_E$}
  \[
    \prftree[r]{$\forall_E$}
    {\Gamma,x:U\vdash\ve t:\sui{n-1}\alpha_i\cdot V_i + \alpha_n \cdot \forall Y.V_n}
    {\Gamma,x:U\vdash\ve t:\sui{n-1}\alpha_i\cdot U_i + \alpha_n\cdot U_n[B/Y]}
  \]
  By the induction hypothesis, $\Gamma\vdash\ve t[\ve
  b/x]:\sui{n-1}\alpha_i\cdot V_i + \alpha_n \cdot \forall Y.V_n$.
  By rule $\forall_E$, $\Gamma\vdash\ve t[\ve
  b/x]:\sui{n-1}\alpha_i\cdot V_i + \alpha_n\cdot V_n[B/Y]$.
  \inductioncase{Case $S$}
  \[
    \prftree[r]{$S$}
    {\Gamma,x:U\vdash\ve{t}:T_i~\forall i \in \{1,\dots,n\}}
    {\Gamma,x:U\vdash \left(\sui{n} \alpha_i\right) \cdot \ve{t}: \sui{n} \alpha_i \cdot T_i}
  \]
  By the
  induction hypothesis, for all $i \in \{1, \dots, n\}$, $\Gamma\vdash\ve t[\ve b/x]:T_i$.
  Then by
  rule $S$, $\Gamma\vdash \left(\sui{n} \alpha_i\right) \cdot \ve{t}[\ve b/x]: \sui{n} \alpha_i \cdot T_i$.
  Notice that $\left(\sui{n} \alpha_i\right) \cdot \ve{t}[\ve b/x]=(\left(\sui{n} \alpha_i\right) \cdot \ve{t})[\ve b/x]$.
  \inductioncase{Case $+_I$}
  \[
    \prftree[r]{$+_I$}
    {\Gamma,x:U\vdash\ve r:R}
    {\Gamma,x:U\vdash\ve u:S}
    {\Gamma,x:U\vdash\ve r+\ve u:R+S}
  \]
  By the
  induction hypothesis, $\Gamma\vdash\ve r[\ve b/x]:R$
  and $\Gamma\vdash\ve u[\ve b/x]:S$.
  Then by
  rule $+_I$, $\Gamma\vdash\ve r[\ve b/x]+\ve u[\ve b/x]:R+S$.
  Notice that $\ve r[\ve b/x]+\ve
  u[\ve b/x]=(\ve r+\ve u)[\ve b/x]$.
  \inductioncase{Case $\equiv$}
  \[
    \prftree[r]{$\equiv$}
    {\Gamma,x:U\vdash\ve t:T}
    {T\equiv R}
    {\Gamma,x:U\vdash\ve t:R}
  \]
  By the induction hypothesis, $\Gamma\vdash\ve t[\ve b/x]:R$.
  Hence, by rule $\equiv$, $\Gamma\vdash\ve t[\ve b/x]:T$.
  \inductioncase{Case $1_E$}
  \[
    \prftree[r]{$1_E$}
    {\Gamma,x:U\vdash1\cdot\ve t:T}
    {\Gamma,x:U\vdash\ve t: T}
  \]
  By the induction hypothesis, $\Gamma\vdash 1 \cdot \ve t[\ve b/x]:R$.
  Hence, by rule $1_E$, $\Gamma\vdash\ve t[\ve b/x]:T$.
\end{proof}

\xrecap{Theorem}{Subject Reduction}{thm:sr}{
  For any terms $\ve{t}, \ve{t}'$, any context $\Gamma$ and any type $T$, if $\ve{t} \to \ve{t}'$ and $\Gamma \vdash \ve{t}: T$, then $\Gamma \vdash \ve{t}': T$.
}
\begin{proof}
  Let $\ve{t} \to \ve{t}'$ and $\Gamma \vdash \ve{t}: T$, we proceed by induction on the rewrite relation:
  \inductioncase{Group E}
  \inductioncase{Case $1\cdot \ve{t}\to \ve{t}$}
  Consider $\Gamma \vdash 1\cdot \ve{t}: T$, then by $1_E$ rule, then $\Gamma \vdash \ve{t}: T$.
  \inductioncase{Case $\alpha\cdot(\beta\cdot\ve{t})\to (\alpha\times\beta)\cdot\ve{t}$}
  Consider $\pi = \Gamma \vdash \alpha\cdot(\beta\cdot\ve{t}): T$,
  then by applying Lemma~\ref{lem:sr:scalars}, there exist $R_1, \dots, R_n$, $\alpha_1, \dots, \alpha_n$ such that
  \begin{itemize}
  \item $T \equiv \sui{n}\alpha_i \cdot R_i$.
  \item $\pi_i = \Gamma \vdash \beta\cdot\ve{t}: R_i$, with $size(\pi) > size(\pi_i)$, for $i \in \{1, \dots, n\}$.
  \item $\sui{n} \alpha_i = \alpha$.
  \end{itemize}
  By applying Lemma~\ref{lem:sr:scalars} for all $i \in \{1,\dots,n\}$,
  there exist $S_{(i,1)}, \dots, S_{(i,m_i)}$, $\beta_{(i,1)}, \dots, \beta_{(i,m_i)}$ such that
  \begin{itemize}
  \item $R_i \equiv \suj{m_i}\beta_{(i,j)} \cdot S_{(i,j)}$.
  \item $\pi_{(i,j)} = \Gamma \vdash \ve{t}: S_{(i,j)}$, with $size(\pi_i) > size(\pi_{(i,j)})$, for $j \in \{1, \dots, m_i\}$.
  \item $\suj{m_i} \beta_{(i,j)} = \beta$.
  \end{itemize}
  Notice that
  \[
    \sui{n}\alpha_i \cdot \underbrace{\suj{m_i}\beta_{(i,j)}}_{\beta} = \sui{n}\alpha_i \cdot \beta\\
    = \beta \cdot \underbrace{\sui{n}\alpha_i}_{\alpha} = \beta \times \alpha = \alpha \times \beta
  \]
  Then applying the $S$ rule, 
  \[
    \prftree[r]{$S$}
    {\Gamma \vdash \ve{t}: S_{(i,j)}~\forall i \in \{1,\dots,n\},~\forall j \in \{1,\dots,m_i\}}
    {\Gamma \vdash (\alpha \times \beta) \cdot \ve{t}: \sui{n}\alpha_i \cdot \suj{m_i}\beta_{(i,j)} \cdot S_{(i,j)}}
  \]
  Since for all $i \in \{1,\dots,n\}$, $\suj{m_i}\beta_{(i,j)} \cdot S_{(i,j)} \equiv R_i$, and since
  $\sui{n}\alpha_i \cdot R_i \equiv T$, then by $\equiv$ rule, we conclude that $\Gamma \vdash (\alpha\times\beta)\cdot\ve{t}:T$.
  \inductioncase{Case $\alpha\cdot(\ve{t} + \ve{r})\to \alpha\cdot\ve{t} + \alpha\cdot\ve{r}$}
  Consider $\Gamma \vdash \alpha\cdot(\ve{t} + \ve{r}): T$,
  then by Lemma~\ref{lem:sr:scalars} there exist $R_1, \dots, R_n$, $\alpha_1, \dots, \alpha_n$ such that
  \begin{itemize}
  \item $T \equiv \sui{n}\alpha_i \cdot R_i$.
  \item $\pi_i = \Gamma \vdash \ve{t} + \ve{r}: R_i$, with $size(\pi) > size(\pi_i)$, for $i \in \{1, \dots, n\}$.
  \item $\sui{n} \alpha_i = \alpha$.
  \end{itemize}
  Since $size(\pi) > size(\pi_i)$, then by Lemma~\ref{lem:sr:sums}, for all $i \in \{1,\dots,n\}$,
  there exist $S_{i,1}, S_{i,2}$ such that
  \begin{itemize}
  \item $\Gamma \vdash \ve{t}: S_{(i,1)}$.
  \item $\Gamma \vdash \ve{r}: S_{(i,2)}$.
  \item $S_{(i,1)} + S_{(i,2)} \equiv R_i$.
  \end{itemize}
  Then applying the $S$ rule, 
  \begin{align*}
    \prftree[r]{$S$}
    {\Gamma \vdash \ve{t}: S_{(i,1)}~\forall i \in \{1,\dots,n\}}
    {\Gamma \vdash \alpha \cdot  \ve{t}: \sui{n} \alpha_i \cdot S_{(i,1)}}
    \qquad
    \prftree[r]{$S$}
    {\Gamma \vdash \ve{r}: S_{(i,2)}~\forall i \in \{1,\dots,n\}}
    {\Gamma \vdash \alpha \cdot \ve{r}: \sui{n} \alpha_i \cdot S_{(i,2)}}
  \end{align*}
  By applying the $+_I$ rule,
  \[
    \prftree[r]{$+_I$}
    {\Gamma \vdash \alpha\cdot\ve{t}: \sui{n} \alpha_i \cdot S_{(i,1)}}
    {\Gamma \vdash \alpha\cdot\ve{r}: \sui{n} \alpha_i \cdot S_{(i,2)}}
    {\Gamma \vdash \alpha\cdot\ve{t} + \alpha\cdot\ve{r}: \sui{n} \alpha_i \cdot S_{(i,1)} + \sui{n} \alpha_i \cdot S_{(i,2)}}
  \]
  Notice that
  \[
    \sui{n} \alpha_i \cdot S_{(i,1)} + \sui{n} \alpha_i \cdot S_{(i,2)} \equiv \sui{n} \alpha_i \cdot (S_{(i,1)} + S_{(i,2)})
    \equiv \sui{n} \alpha_i \cdot R_i \equiv T
  \]
  Finally, applying the $\equiv$ rule, we conclude that $\Gamma \vdash \alpha\cdot\ve{t} + \alpha\cdot\ve{r}: T$.

  \inductioncase{Group F}
  \inductioncase{Case $\alpha\cdot\ve{t} + \beta\cdot\ve{t}\to (\alpha + \beta)\cdot \ve{t}$}
  Consider ${\Gamma \vdash \alpha\cdot\ve{t} + \beta\cdot\ve{t}: T}$.\\
  For simplicity, we rename $\alpha = \mu_1$ and $\beta = \mu_2$, then by Lemma~\ref{lem:sr:sums} there exist $S_1, S_2$ such that
  \begin{itemize}
  \item ${\pi_1 = \Gamma \vdash \mu_1\cdot\ve{t}: S_{1}}$.
  \item ${\pi_2 = \Gamma \vdash \mu_2\cdot\ve{t}: S_{2}}$.
  \item ${S_{1} + S_{2} \equiv T}$.
  \end{itemize}
  And by Lemma~\ref{lem:sr:scalars}, for $k = 1, 2$,
  there exist $R_{(k,1)},\dots,R_{(k,n_k)}$, $\gamma_{(k,1)},\dots,\gamma_{(k,n_k)}$ such that
  \begin{itemize}
  \item $S_k \equiv \sui{n_k}\gamma_{(k,i)} \cdot R_{(k,i)}$.
  \item $\pi_{(k,i)} = \Gamma \vdash \ve{t}: R_{(k,i)}$, with $size(\pi_k) > size(\pi_{(k,i)})$, for $i \in \{1, \dots, n_k\}$.
  \item $\sui{n_k} \gamma_{(k,i)} = \mu_k$.
  \end{itemize}
  Notice that
  \[
    \underbrace{\sui{n_1} \mu_{(1,i)}}_{=~\mu_1} + \underbrace{\sui{n_2} \mu_{(2,i)}}_{=~\mu_2} = \mu_1 + \mu_2 = \alpha + \beta
  \]
  Then applying the $S$ rule, 
  \[
    \prftree[r]{$S$}
    {\Gamma \vdash \ve{t}: R_{(1,i)}~\forall i \in \{1,\dots,n_1\}}
    {\Gamma \vdash \ve{t}: R_{(2,i)}~\forall i \in \{1,\dots,n_2\}}
    {\Gamma \vdash (\alpha + \beta) \cdot \ve{t}: \sui{n_1} \mu_{(1,i)} \cdot R_{(1,i)} + \sui{n_2} \mu_{(2,i)} \cdot R_{(2,i)}}
  \]
  We also know that
  \[
    \sui{n_1} \mu_{(1,i)} \cdot R_{(1,i)} \equiv S_1\qquad
    \sui{n_2} \mu_{(2,i)} \cdot R_{(2,i)} \equiv S_2\qquad
    S_1 + S_2 \equiv T
  \]
  Finally, we conclude by $\equiv$ rule that ${\Gamma \vdash (\alpha + \beta)\cdot \ve{t}: T}$.
  \inductioncase{Case $\alpha\cdot\ve{t} + \ve{t}\to (\alpha + 1)\cdot \ve{t}$}
  Consider ${\Gamma \vdash \alpha\cdot\ve{t} + \ve{t}: T}$, then by Lemma~\ref{lem:sr:sums} there exist $S_1, S_2$ such that
  \begin{itemize}
  \item ${\pi = \Gamma \vdash \alpha\cdot\ve{t}: S_{1}}$.
  \item ${\Gamma \vdash \ve{t}: S_{2}}$.
  \item ${S_{1} + S_{2} \equiv T}$.
  \end{itemize}
  And by Lemma~\ref{lem:sr:scalars}, there exist $R_1, \dots, R_n$, $\alpha_1, \dots, \alpha_n$ such that
  \begin{itemize}
  \item $S_1 \equiv \sui{n}\alpha_{i} \cdot R_i$.
  \item $\pi_{i} = \Gamma \vdash \ve{t}: R_{i}$, with $size(\pi) > size(\pi_{i})$, for $i \in \{1, \dots, n\}$.
  \item $\sui{n} \alpha_{i} = \alpha$.
  \end{itemize}
  Then applying the $S$ rule, 
  \[
    \prftree[r]{$S$}
    {\Gamma \vdash \ve{t}: R_i~\forall i \in \{1,\dots,n\}}
    {\Gamma \vdash \ve{t}: S_2}
    {\Gamma \vdash (\alpha + 1) \cdot \ve{t}: \sui{n}\alpha_{i} \cdot R_i + S_{2}}
  \]
  We also know that
  \[
    \sui{n} \mu_{i} \cdot R_{i} \equiv S_1\qquad
    S_1 + S_2 \equiv T
  \]
  Finally, we conclude by $\equiv$ rule that ${\Gamma \vdash (\alpha + 1)\cdot \ve{t}: T}$.
  \inductioncase{Case $\ve{t} + \ve{t}\to (1 + 1)\cdot\ve{t}$}
  Consider $\Gamma \vdash \ve{t} + \ve{t}: T$, then
  by Lemma~\ref{lem:sr:sums} there exist $T_1, T_2$ such that
  \begin{itemize}
  \item $\Gamma \vdash \ve{t}: T_1$.
  \item $\Gamma \vdash \ve{t}: T_2$.
  \item $T_1 + T_2 \equiv T$.
  \end{itemize}
  Then applying the $S$ rule, 
  \[
    \prftree[r]{$S$}
    {\Gamma \vdash \ve{t}: T_1}
    {\Gamma \vdash \ve{t}: T_2}
    {\Gamma \vdash (1 + 1)\cdot\ve{t}: T_1 + T_2}
  \]
  Finally, by $\equiv$ rule we conclude that ${\Gamma \vdash (1 + 1)\cdot \ve{t}: T}$.

  \inductioncase{Group B}
  \inductioncase{Case $(\lambda x.\ve{t})~\ve{b}\to\ve{t}\subst{\ve b}{x}$}
  Consider $\Gamma \vdash (\lambda x.\ve{t})~\ve{b}: T$, then by Lemma~\ref{lem:sr:app}, , there exist $R_1, \dots, R_h$, $\mu_1, \dots, \mu_h$, $\V_1,\dots,\V_h$ such that $T \equiv \suk{h} \mu_k \cdot R_k$,
  $\suk{h} \mu_k = 1$ and for all $k \in \{1,\dots,h\}$,
  \begin{itemize}
  \item $\Gamma\vdash \lambda x.\ve{t}: \sui{n_k}{\alpha_{(k,i)} \cdot\forall\vec{X}.(U\to T_{(k,i)})}$.
  \item $\Gamma\vdash\ve{b}: \suj{m_k}\beta_{(k,j)}\cdot U[\vec{A}_{(k,j)}/\vec{X}]$.
  \item $\sui{n_k}\suj{m_k} \alpha_{(k,i)}\times\beta_{(k,j)}\cdot {T_{(k,i)}[\vec{A}_{(k,j)}/\vec{X}]} \ssubt_{\V_k,\Gamma} R_k$.
  \end{itemize}
  For the sake of readability, we will split the proof:
  \begin{enumerate}
	\item We will prove that $\Gamma, x:U[\vec{A}_{(k,j)}/X] \vdash \ve{t}: T_{(k,i)}[\vec{A}_{(k,j)}/X]$, for all $k \in \{1, \dots, h\}$, $j \in \{1, \dots, m\}$, $i \in \{1, \dots, n\}$.
	\item We will prove that $\Gamma \vdash \ve{t}\subst{\ve b}{x}: T_{(k,i)}[\vec{A}_{(k,j)}/X]$, for all $k \in \{1, \dots, h\}$, $j \in \{1, \dots, m\}$, $ i \in \{1, \dots, n\}$.
	\item We will prove that $\Gamma \vdash \ve{t}\subst{\ve b}{x}: T$.
  \end{enumerate}
  \medskip
  \textleadbydots{Item (1)}
  \noindent We will prove that $\Gamma, x:U[\vec{A}_{(k,j)}/X] \vdash \ve{t}: T_{(k,i)}[\vec{A}_{(k,j)}/X]$, for all $k \in \{1, \dots, h\}$, $j \in \{1, \dots, m\}$, $i \in \{1, \dots, n\}$.\\
  For simplicity, we will omit the $k$ index, which would otherwise be present in all the types, scalars and upper bound of the summations.\\
  Considering $\lambda x.\ve{t}$ is a basis term, by Lemma~\ref{lem:sr:basevectors} then
  there exist $W_1,\dots,W_b$, $\gamma_1,\dots,\gamma_b$ such that
  \begin{itemize}
  \item $\sug{a}{b} \gamma_{a} \cdot W_{a} \equiv \sui{n}{\alpha_{i} \cdot\forall\vec{X}.(U\to T_{i})}$.
  \item $\Gamma \vdash \lambda x.\ve{t}: W_{a}$, for $a \in \{1, \dots,b\}$.
  \item $\sug{a}{b} \gamma_{a} = 1$.
  \end{itemize}
  Without loss of generality, we assume that all unit types present at both sides of the equivalences are distinct, so
  by Lemma~\ref{lem:sr:equivdistinctscalars}, then $b = n$ and there exists a permutation of $n$, $p$, such that
  $\forall\vec{X}.(U\to T_{i}) \equiv W_{p(i)}$ and $\alpha_{i} = \gamma_{p(i)}$, for all $i \in \{1, \dots, n\}$.\\
  Since for all $i \in \{1, \dots, n\}$ we have $\Gamma \vdash \lambda x.\ve{t}: \forall\vec{X}.(U\to T_{i})$,
  then by Lemma~\ref{lem:sr:abs} and Lemma~\ref{lem:sr:equivdistinctscalars},
  we know that $\Gamma, x:V_{i} \vdash \ve{t}: S_{i}$, and $V_{i} \to S_{i} \ssubt_{\V_i,\Gamma} \forall\vec{X}.(U\to T_{i})$.\\
  By applying Lemma~\ref{lem:sr:arrowscomp}, then $U\equiv V_{i}[\vec{B}/\vec{Y}]$ and $T_{i} \equiv S_{i}[\vec{B}/\vec{Y}]$, with $\vec{Y} \notin \FV{\Gamma}$.\\
  Then, by Lemma~\ref{lem:sr:substitution} and $\equiv$ rule, we have that $\Gamma, x:U \vdash \ve{t}: T_{i}$ for all $i \in \{1, \dots, n\}$.\\
  By Lemma~\ref{lem:sr:sorderhasnofv}, since $V_{i} \to S_{i} \ssubt_{\V_i,\Gamma} \forall\vec{X}.(U\to T_{i})$ for all $i \in \{1, \dots, n\}$, then we know $\vec{X} \notin \FV{\Gamma}$
  and so  $\Gamma \equiv \Gamma[\vec{C}/\vec{X}]$, for any $\vec{C}$.\\
  Therefore, by applying Lemma~\ref{lem:sr:substitution} multiple times, we have $\Gamma, x:U[\vec{A}_{j}/X] \vdash \ve{t}: T_{i}[\vec{A}_{j}/X]$
  for all $j \in \{1, \dots, m\}$, $i \in \{1, \dots, n\}$.\\
  Following this procedure for all $k \in \{1, \dots, h\}$, then we proved
  that $\Gamma, x:U[\vec{A}_{(k,j)}/X] \vdash \ve{t}: T_{(k,i)}[\vec{A}_{(k,j)}/X]$,
  for all $k \in \{1, \dots, h\}$, $j \in \{1, \dots, m\}$, $i \in \{1, \dots, n\}$.

  \medskip
  \textleadbydots{Item (2)}
  \noindent We will prove that $\Gamma \vdash \ve{t}\subst{\ve b}{x}: T_{(k,i)}[\vec{A}_{(k,j)}/X]$, for all $k \in \{1, \dots, h\}$, $j \in \{1, \dots, m_k\}$, $ i \in \{1, \dots, n_k\}$.\\
  For simplicity, we will omit the $k$ index, which would otherwise be present in all the types, scalars and upper bound of the summations.\\
  Since $\ve{b}$ is a basis term, by Lemma~\ref{lem:sr:basevectors} there exist $W'_1, \dots, W'_c$, $\eta_1, \dots, \eta_c$ such that
  \begin{itemize}
  \item $\sug{a}{c} \eta_{a} \cdot W'_{a} \equiv \suj{m}\beta_j\cdot U[\vec{A}_{j}/\vec{X}]$.
  \item $\Gamma \vdash \ve{b}: W'_{a}$, for $a \in \{1, \dots,c\}$.
  \item $\sug{a}{c} \eta_{a} = 1$.
  \end{itemize}
  Without loss of generality, we assume that all unit types present at both sides of the equivalences are distinct, so
  by Lemma~\ref{lem:sr:equivdistinctscalars}, then $c = m$, and there exists a permutation $q$ of $m$, such that
  $U[\vec{A}_{j}/\vec{X}] \equiv W'_{q(j)}$ and $\beta_j = \eta_{q(j)}$, for all $j \in \{1, \dots, m\}$.\\
  Then, following Item (1), by applying Lemma~\ref{lem:sr:substitution}, we have that $\Gamma \vdash \ve{t}\subst{\ve b}{x}: T_{i}[\vec{A}_{j}/X]$
  for all $j \in \{1, \dots, m\}$, $i \in \{1, \dots, n\}$.
  Following this procedure for all $k \in \{1, \dots, h\}$, then we proved
  that $\Gamma \vdash \ve{t}\subst{\ve b}{x}: T_{(k,i)}[\vec{A}_{(k,j)}/X]$,
  for all $k \in \{1, \dots, h\}$, $j \in \{1, \dots, m\}$, $ i \in \{1, \dots, n\}$.

  \medskip
  \textleadbydots{Item (3)}
  \noindent Using the results of Item (1) and Item (2), and since in both items we already proved
  that for all $k \in \{1, \dots, h\}$, $\sui{n_k} \alpha_i = \suj{m_k} \beta_j = 1$,
  then by applying the $S$ rule for all $k \in \{1, \dots, h\}$ (we will omit the $k$ index for simplicity, that
  will be present in all types, scalars and upper bound of the summations), 
  \[
    \prftree[r]{$1_E$}
    {\prftree[r]{$S$}
      {\Gamma \vdash \ve{t}\subst{\ve b}{x}: T_{i}[\vec{A}_{j}/X]~\forall i \in \{1,\dots,n\},~\forall j \in \{1,\dots,m\}}
      {\Gamma \vdash 1\cdot \ve{t}\subst{\ve b}{x}: \sui{n}\suj{m}\alpha_i\times\beta_j\cdot T_i[\vec{A}_j/X]}}
    {\Gamma \vdash \ve{t}\subst{\ve b}{x}: \sui{n}\suj{m}\alpha_i\times\beta_j\cdot T_i[\vec{A}_j/X]}
  \]
  Since $\sui{n_k}\suj{m_k}\alpha_{(k,i)}\times\beta_{(k,j)}\cdot T_{(k,i)}[\vec{A}_{(k,j)}/X] \ssubt_{\V,\Gamma} R_k$,
  then $\Gamma \vdash \ve{t}\subst{\ve b}{x}: R_k$.\\
  Considering that $\suk{h} \mu_k = 1$, then by applying the $S$ and the $1_E$ rule again, 
  \[
    \prftree[r]{$1_E$}
    {\prftree[r]{$S$}
      {\Gamma \vdash \ve{t}\subst{\ve b}{x}: R_k~\forall k \in \{1,\dots,h\}}
      {\Gamma \vdash 1\cdot\ve{t}\subst{\ve b}{x}: \suk{h} \mu_k \cdot R_k}}
    {\Gamma \vdash \ve{t}\subst{\ve b}{x}: \suk{h} \mu_k \cdot R_k}
  \]
  Finally, since $\mu_k \cdot R_k \equiv T$,
  we conclude by $\equiv$ rule that $\Gamma \vdash \ve{t}[\ve{b}/x]: T$.

  \inductioncase{Group A}
  \inductioncase{Case $(\ve{t} + \ve{r})~\ve{u}\to (\ve{t})~\ve{u} + (\ve{r})~\ve{u}$}
  Consider $\Gamma \vdash (\ve{t} + \ve{r})~\ve{u}: T$, then by
  Lemma~\ref{lem:sr:app}, there exist $R_1, \dots, R_h$, $\mu_1, \dots, \mu_h$, $\V_1,\dots,\V_h$ such that $T \equiv \suk{h} \mu_k \cdot R_k$,
  $\suk{h} \mu_k = 1$ and for all $k \in \{1,\dots,h\}$
  \begin{itemize}
  \item $\Gamma\vdash\ve{t} + \ve{r}: \sui{n_k}{\alpha_{(k,i)} \cdot\forall\vec{X}.(U\to T_{(k,i)})}$.
  \item $\Gamma\vdash\ve{u}: \suj{m_k}\beta_{(k,j)}\cdot U[\vec{A}_j/\vec{X}]$.
  \item $\sui{n_k}\suj{m_k} \alpha_{(k,i)}\times\beta_{(k,j)}\cdot {T_{(k,i)}[\vec{A}_{(k,j)}/\vec{X}]} \ssubt_{\V_k,\Gamma} R_k$.
  \end{itemize}
  We will simplify the rest of this proof by omitting the $k$ index,
  which would otherwise be present in all the types, scalars and upper bound of the summations.
  The rest of this proof then should be applied to all $k \in \{1, \dots, h\}$.\\
  By Lemma~\ref{lem:sr:sums}, there exist $S_1$, $S_2$ such that
  \begin{itemize}
  \item $\Gamma \vdash \ve{t}: S_{1}$.
  \item $\Gamma \vdash \ve{r}: S_{2}$.
  \item $S_{1} + S_{2} \equiv \sui{n}{\alpha_i \cdot\forall\vec{X}.(U\to T_i)}$.
  \end{itemize}
  Hence, there exist $N_1, N_2\subseteq\{1,\dots,n\}$ with $N_1\cup N_2=\{1,\dots,n\}$ such that
  \begin{align*}
    S_{1}\equiv\sum\limits_{i\in N_1\setminus N_2}\alpha_{i}\cdot\forall\vec X.(U\to T_{i})+
    \sum\limits_{i\in N_1\cap N_2}\eta_{i}\cdot\forall\vec X.(U\to T_{i}) & \mbox{\quad and}\\
    S_{2}\equiv\sum\limits_{i\in N_2\setminus N_1}\alpha_{i}\cdot\forall\vec X.(U\to T_{i})+
    \sum\limits_{i\in N_1\cap N_2}\eta'_{i}\cdot\forall\vec X.(U\to T_{i}) &
  \end{align*}
  where for all $i \in N_1\cap N_2$, $\eta_{i}+\eta'_{i}=\alpha_{i}$.
  Therefore, using $\equiv$ we get
  \begin{align*}
    \Gamma\vdash\ve t:\sum\limits_{i\in N_1\setminus N_2}\alpha_{i}\cdot\forall\vec X.(U\to T_{i})+
    \sum\limits_{i\in N_1\cap N_2}\eta_{i}\cdot\forall\vec X.(U\to T_{i}) & \mbox{\quad and}\\
    \Gamma\vdash\ve r:\sum\limits_{i\in N_2\setminus N_1}\alpha_{i}\cdot\forall\vec X.(U\to T_{i})+
    \sum\limits_{i\in N_1\cap N_2}\eta'_{i}\cdot\forall\vec X.(U\to T_{i}) &
  \end{align*}
  So, using rule $\to_E$, we get
  \begin{align*}
    \Gamma\vdash(\ve t)~\ve u:\sum\limits_{i\in N_1\setminus N_2}\suj{m}\alpha_{i}\times\beta_{j}\cdot T_{i}[\vec{A}_{j}/\vec X]+
    \sum\limits_{i\in N_1\cap N_2}\suj{m}\eta'_{i}\times\beta_{j}\cdot T_{i}[\vec{A}_{j}/\vec X] & \mbox{\quad and}\\
    \Gamma\vdash(\ve r)~\ve u:\sum\limits_{i\in N_2\setminus N_1}\suj{m}\alpha_{i}\times\beta_{j}\cdot T_{i}[\vec{A}_{j}/\vec X]+
    \sum\limits_{i\in N_1\cap N_2}\suj{m}\eta'_{i}\times\beta_{j}\cdot T_{i}[\vec{A}_{j}/\vec X] &
  \end{align*}
  By rule $+_I$ we can conclude
  \[
    \Gamma\vdash(\ve t)~\ve u+(\ve r)~\ve u:\sui{n}\suj{m}\alpha_{i}\times\beta_{j}\cdot T_{i}[\vec{A}_{j}/\vec X]
  \]
  Since $\sui{n_k}\suj{m_k}\alpha_{(k,i)}\times\beta_{(k,j)}\cdot T_{(k,i)}[\vec{A}_{(k,j)}/\vec X] \ssubt_{\V_k, \Gamma} R_k$
  for all $k \in \{1, \dots, h\}$, then by definition of $\ssubt$, we can derive $\Gamma\vdash(\ve t)~\ve u+(\ve r)~\ve u: R_k$.\\
  By applying the $S$ and $1_E$ rules, then 
  \[
    \prftree[r]{$1_E$}
    {\prftree[r]{$S$}
      {\Gamma\vdash(\ve t)~\ve u+(\ve r)~\ve u: R_k~\forall k \in \{1,\dots,k\}}
      {\Gamma\vdash 1\cdot((\ve t)~\ve u+(\ve r)~\ve u): \suk{h} \mu_k \cdot R_k}}
    {\Gamma\vdash (\ve t)~\ve u+(\ve r)~\ve u: \suk{h} \mu_k \cdot R_k}
  \]
  Finally, by the $\equiv$ rules, then $\Gamma\vdash (\ve t)~\ve u+(\ve r)~\ve u: T$.
  \inductioncase{Case $(\ve{t})~(\ve{r} + \ve{u})\to (\ve{t})~\ve{r} + (\ve{t})~\ve{u}$}
  Consider $\Gamma \vdash (\ve{t})~(\ve{r} + \ve{u}): T$, then by
  Lemma~\ref{lem:sr:app}, there exist $R_1, \dots, R_h$, $\mu_1, \dots, \mu_h$, $\V_1,\dots,\V_h$ such that $T \equiv \suk{h} \mu_k \cdot R_k$,
  $\suk{h} \mu_k = 1$ and for all $k \in \{1,\dots,h\}$
  \begin{itemize}
  \item $\Gamma\vdash\ve{t}: \sui{n_k}{\alpha_{(k,i)} \cdot\forall\vec{X}.(U\to T_{(k,i)})}$.
  \item $\Gamma\vdash \ve{r} + \ve{u}: \suj{m_k}\beta_{(k,j)}\cdot U[\vec{A}_{(k,j)}/\vec{X}]$.
  \item $\sui{n_k}\suj{m_k} \alpha_{(k,i)}\times\beta_{(k,j)}\cdot {T_{(k,i)}[\vec{A}_{(k,j)}/\vec{X}]} \ssubt_{\V_k,\Gamma} R_k$.
  \end{itemize}
  We will simplify the rest of this proof by omitting the $k$ index,
  which would otherwise be present in all the types, scalars and upper bound of the summations.
  The rest of this proof then should be applied to all $k \in \{1, \dots, h\}$.\\
  By Lemma~\ref{lem:sr:sums}, there exists $S_1$, $S_2$ such that
  \begin{itemize}
  \item $\Gamma \vdash \ve{r}: S_{1}$
  \item $\Gamma \vdash \ve{u}: S_{2}$
  \item $S_{1} + S_{2} \equiv \suj{m}\beta_j\cdot U[\vec{A}_j/\vec{X}]$
  \end{itemize}
  Hence, there exist $N_1, N_2\subseteq\{1,\dots,m\}$ with $N_1\cup N_2=\{1,\dots,m\}$,
  such that
  \begin{align*}
    S_1\equiv\sum\limits_{j\in N_1\setminus N_2}\beta_j\cdot U[\vec{A}_j/\vec{X}]+
    \sum\limits_{j\in N_1\cap N_2}\eta_{j}\cdot U[\vec{A}_j/\vec{X}] & \mbox{\quad and}\\
    S_2\equiv\sum\limits_{i\in N_2\setminus N_1}\beta_j\cdot U[\vec{A}_j/\vec{X}]+
    \sum\limits_{j\in N_1\cap N_2}\eta'_{kj}\cdot U[\vec{A}_j/\vec{X}] &
  \end{align*}
  where for all $j\in N_1\cap N_2$, $\eta_{j}+\eta'_{j}=\beta_j$.
  Therefore, using $\equiv$ we get
  \begin{align*}
    \Gamma\vdash\ve r:\sum\limits_{j\in N_1\setminus N_2}\beta_j\cdot U[\vec{A}_j/\vec{X}]+
    \sum\limits_{j\in N_1\cap N_2}\eta_{j}\cdot U[\vec{A}_j/\vec{X}] & \mbox{\quad and}\\
    \Gamma\vdash\ve u:\sum\limits_{j\in N_2\setminus N_1}\beta_j\cdot U[\vec{A}_j/\vec{X}]+
    \sum\limits_{j\in N_1\cap N_2}\eta'_{kj}\cdot U[\vec{A}_j/\vec{X}] &
  \end{align*}
  So, using rule $\to_E$, we get
  \begin{align*}
    \Gamma\vdash(\ve t)~\ve r:\sui{n}\sum\limits_{j\in N_1\setminus N_2}\alpha_i\times\beta_j\cdot T_i[\vec{A}_j/\vec X]+
    \sui{n}\sum\limits_{j\in N_1\cap N_2}\alpha_i\times\eta_{j}\cdot T_i[\vec{A}_j/\vec X] & \mbox{\quad and}\\
    \Gamma\vdash(\ve t)~\ve u:\sui{n}\sum\limits_{j\in N_2\setminus N_1}\alpha_i\times\beta_j\cdot T_i[\vec{A}_j/\vec X]+
    \sui{n}\sum\limits_{j\in N_1\cap N_2}\alpha_i\times\eta'_{kj}\cdot T_i[\vec{A}_j/\vec X] &
  \end{align*}
  By rule $+_I$ we can conclude
  \[
    \Gamma\vdash(\ve t)~\ve r+(\ve t)~\ve u:\sui{n}\suj{m}\alpha_i\times\beta_j\cdot T_i[\vec{A}_j/\vec X]
  \]
  Since $\sui{n_k}\suj{m_k}\alpha_{(k,i)}\times\beta_{(k,j)}\cdot T_{(k,i)}[\vec{A}_{(k,j)}/\vec X] \ssubt_{\V_k, \Gamma} R_k$
  for all $k \in \{1, \dots, h\}$, then by definition of $\ssubt$, we can derive $\Gamma\vdash(\ve t)~\ve r+(\ve t)~\ve u: R_k$.\\
  By applying the $S$ and $1_E$ rules, then 
  \[
    \prftree[r]{$1_E$}
    {\prftree[r]{$S$}
      {\Gamma\vdash(\ve t)~\ve r+(\ve t)~\ve u: R_k~\forall k \in \{1,\dots,h\}}
      {\Gamma\vdash 1\cdot((\ve t)~\ve r+(\ve t)~\ve u): \suk{h} \mu_k \cdot R_k}}
    {\Gamma\vdash (\ve t)~\ve r+(\ve t)~\ve u: \suk{h} \mu_k \cdot R_k}
  \]
  Finally, by the $\equiv$ rules, then $\Gamma\vdash (\ve t)~\ve r+(\ve t)~\ve u: T$.
  \inductioncase{Case $(\alpha\cdot\ve{t})~\ve{r}\to \alpha\cdot(\ve{t})~\ve{r}$}
  Consider $\Gamma \vdash (\alpha\cdot\ve{t})\ \ve{r}: T$, by
  Lemma~\ref{lem:sr:app}, there exist $R_1, \dots, R_h$, $\mu_1, \dots, \mu_h$, $\V_1,\dots,\V_h$ such that $T \equiv \suk{h} \mu_k \cdot R_k$,
  $\suk{h} \mu_k = 1$ and for all $k \in \{1,\dots,h\}$
  \begin{itemize}
  \item $\pi_k = \Gamma\vdash \alpha\cdot\ve{t}: \sui{n_{k}}{\alpha_{(k,i)} \cdot\forall\vec{X}.(U\to T_{(k,i)})}$.
  \item $\Gamma\vdash \ve{r}: \suj{m_{k}}\beta_{(k,j)}\cdot U[\vec{A}_{(k,j)}/\vec{X}]$.
  \item $\sui{n_{k}}\suj{m_{k}} \alpha_{(k,i)}\times\beta_{(k,j)}\cdot {T_{(k,i)}[\vec{A}_{(k,j)}/\vec{X}]} \ssubt_{\V_{k},\Gamma} R_k$.
  \end{itemize}
  We will simplify the rest of this proof by omitting the $k$ index,
  which would otherwise be present in all the types, scalars and upper bound of the summations.
  The rest of this proof then should be applied to all $k \in \{1, \dots, h\}$.\\
  By Lemma~\ref{lem:sr:scalars}, there exist $S_1, \dots, S_b$, $\eta_1, \dots, \eta_b$ such that
  \begin{itemize}
  \item $\sui{n}{\alpha_{i} \cdot\forall\vec{X}.(U\to T_{i})} \equiv \sug{a}{b}\eta_{a} \cdot S_{a}$.
  \item $\pi_i = \Gamma \vdash \ve{t}: S_{a}$, with $size(\pi) > size(\pi_{a})$, for $a \in \{1, \dots, b\}$.
  \item $\sug{a}{b} \eta_{a} = \alpha$.
  \end{itemize}
  Considering $\sui{n}{\alpha_{i} \cdot\forall\vec{X}.(U\to T_{i})}$ does not have any general variable $\vara{X}$ and
  that $\sui{n}{\alpha_{i} \cdot\forall\vec{X}.(U\to T_{i})} \equiv \sug{a}{b}\eta_{a} \cdot S_{a}$,
  then by Lemma~\ref{lem:sr:typecharact},
  $S_{a} \equiv \sug{c}{d_{a}}\gamma_{(a,c)}\cdot V_{(a,c)}$.\\
  Without loss of generality, we assume that all unit types present at both sides of the equivalences are distinct, so
  by Lemma~\ref{lem:sr:equivdistinctscalars}, then $n = \sug{a}{b} d_{a}$, and by taking a partition
  from $\{1,\dots, \sug{a}{b} d_{a}\}$ (defining an equivalence class)
  and the trivial permutation $p$ of $n$ such that $p(i) = i$ (which we will omit for readability), we have
  \begin{itemize}
  \item $\alpha_i = \eta_{[i]}\times\sigma_i$, where $\sigma_i = \gamma_{\left([i],\frac{i}{[i]}\right)}$.
  \item $\forall\vec{X}.(U\to T_i) \equiv V_{\left([i],\frac{i}{[i]}\right)}$.
  \end{itemize}
  Take $f(a) = \sug{e}{a-1} d_e$, so we rewrite $S_a \equiv \sug{c}{d_{a}}\gamma_{(a,c)}\cdot V_{(a,c)}$ as
  \[
    S_a \equiv \sum^{f(a) + d_a}_{g = f(a)} \sigma_g\cdot V_{\left([g],\frac{g}{[g]}\right)}
    \equiv \sum^{f(a) + d_a}_{g = f(a)} \sigma_g\cdot \forall\vec{X}.(U\to T_g)
  \]
  Applying $\to_E$ for all $a \in \{1,\dots,b\}$,
  \[
    \prftree[r]{$\to_E$}
    {\Gamma \vdash \ve{t}: \sum^{f(a) + d_a}_{g = f(a)} \sigma_g \cdot \forall\vec{X}.(U\to T_g)}
    {\Gamma\vdash \ve{r}: \suj{m}\beta_{j}\cdot U[\vec{A}_{j}/\vec{X}]}
    {\Gamma \vdash (\ve{t})~\ve{r}: \sum^{f(a) + d_a}_{g = f(a)}\suj{m} \left(\sigma_g\times\beta_{j}\right)\cdot T_{g}[\vec{A}_{j}/\vec{X}]}
  \]
  We rewrite $\sum^{f(a) + d_a}_{g = f(a)}\suj{m} \left(\sigma_g\times\beta_{j}\right)\cdot T_{g}[\vec{A}_{j}/\vec{X}] \equiv P_a$,
  then by applying the $S$ rule we have
  \vspace{1cm}
  \[
    \prftree[r]{$S$}
    {\Gamma \vdash (\ve{t})~\ve{r}: P_a~\forall a \in \{1,\dots,b\}}
    {\Gamma \vdash \alpha\cdot(\ve{t})~\ve{r}: \sug{a}{b} \eta_a \cdot P_a}
  \]
  Now we begin to unravel the final result
  \begin{align*}
    \sug{a}{b} \eta_a \cdot P_a &\equiv \sug{a}{b} \eta_a \cdot \sum^{f(a) + d_a}_{g = f(a)}\suj{m} \left(\sigma_g\times\beta_{j}\right)\cdot T_{g}[\vec{A}_{j}/\vec{X}]\\
                                &\equiv \sug{a}{b}\sum^{f(a) + d_a}_{g = f(a)}\suj{m} \left(\eta_{[g]} \times \sigma_g\times\beta_{j}\right)\cdot T_{g}[\vec{A}_{j}/\vec{X}]\\
                                &\equiv \sug{a}{b}\sum^{f(a) + d_a}_{g = f(a)}\suj{m} (\alpha_{g}\times\beta_{j})\cdot T_{g}[\vec{A}_{j}/\vec{X}]\\
                                &\equiv \sui{n}\suj{m} (\alpha_{i}\times\beta_{j})\cdot T_{i}[\vec{A}_{j}/\vec{X}]\\
  \end{align*}
  Then,
  \[
    \Gamma \vdash \alpha\cdot(\ve{t})~\ve{r}: \sui{n}\suj{m} (\alpha_{i}\times\beta_{j})\cdot T_{i}[\vec{A}_{j}/\vec{X}]
  \]
  Since $\sui{n_k}\suj{m_k} (\alpha_{(k,i)}\times\beta_{(k,j)})\cdot T_{(k,i)}[\vec{A}_{(k,j)}/\vec{X}] \ssubt_{\V_k,\Gamma} R_k$,
  then for all $k \in \{1,\dots,h\}$, $\Gamma \vdash \alpha\cdot(\ve{t})~\ve{r}: R_k$.\\
  By applying the $S$ and $1_E$ rules, then 
  \[
    \prftree[r]{$1_E$}
    {\prftree[r]{$S$}
      {\Gamma\vdash\alpha \cdot(\ve t)~\ve r: R_k~\forall k \in \{1,\dots,h\}}
      {\Gamma\vdash 1\cdot(\alpha \cdot(\ve t)~\ve r): \suk{h} \mu_k \cdot R_k}}
    {\Gamma\vdash \alpha \cdot(\ve t)~\ve r: \suk{h} \mu_k \cdot R_k}
  \]
  Finally, by the $\equiv$ rule, then $\Gamma\vdash \alpha \cdot (\ve t)~\ve r: T$.
  \inductioncase{Case $(\ve{t})~(\alpha\cdot\ve{r})\to \alpha\cdot(\ve{t})~\ve{r}$}
  Consider $\Gamma \vdash (\ve{t})~(\alpha\cdot\ve{r}): T$, by
  Lemma~\ref{lem:sr:app}, there exist $R_1, \dots, R_h$, $\mu_1, \dots, \mu_h$, $\V_1,\dots,\V_h$ such that $T \equiv \suk{h} \mu_k \cdot R_k$,
  $\suk{h} \mu_k = 1$ and for all $k \in \{1,\dots,h\}$
  \begin{itemize}
  \item $\Gamma\vdash \ve{t}: \sui{n_{k}}{\alpha_{(k,i)} \cdot\forall\vec{X}.(U\to T_{(k,i)})}$.
  \item $\pi_k = \Gamma\vdash \alpha\cdot\ve{r}: \suj{m_{k}}\beta_{(k,j)}\cdot U[\vec{A}_{(k,j)}/\vec{X}]$.
  \item $\sui{n_{k}}\suj{m_{k}} \alpha_{(k,i)}\times\beta_{(k,j)}\cdot {T_{(k,i)}[\vec{A}_{(k,j)}/\vec{X}]} \ssubt_{\V_k,\Gamma} R_k$.
  \end{itemize}
  We will simplify the rest of this proof by omitting the $k$ index,
  which would otherwise be present in all the types, scalars and upper bound of the summations.
  The rest of this proof then should be applied to all $k \in \{1, \dots, h\}$.\\
  By Lemma~\ref{lem:sr:scalars}, there exist $S_1, \dots, S_b$, $\eta_1, \dots, \eta_b$ such that
  \begin{itemize}
  \item $\suj{m}\beta_{j}\cdot U[\vec{A}_{j}/\vec{X}] \equiv \sug{a}{b}\eta_{a} \cdot S_{a}$.
  \item $\pi_i = \Gamma \vdash \ve{r}: S_{a}$, with $size(\pi) > size(\pi_{a})$, for $a \in \{1, \dots, b\}$.
  \item $\sug{a}{b} \eta_{a} = \alpha$.
  \end{itemize}
  Considering $\suj{m}\beta_{j}\cdot U[\vec{A}_{j}/\vec{X}]$ does not have any general variable $\vara{X}$ and
  that $\suj{m}\beta_{j}\cdot U[\vec{A}_{j}/\vec{X}] \equiv \sug{a}{b}\eta_{a} \cdot S_{a}$,
  then by Lemma~\ref{lem:sr:typecharact},
  $S_{a} \equiv \sug{c}{d_{a}}\gamma_{(a,c)}\cdot V_{(a,c)}$.\\
  Without loss of generality, we assume that all unit types present at both sides of the equivalences are distinct, so
  by Lemma~\ref{lem:sr:equivdistinctscalars}, then $m = \sug{a}{b} d_{a}$, and by taking a partition
  from $\{1,\dots, \sug{a}{b} d_{a}\}$ (defining an equivalence class)
  and the trivial permutation $p$ of $m$ such that $p(j) = j$ (which we will omit for readability), we have
  \begin{itemize}
  \item $\beta_j = \eta_{[j]}\times\sigma_j$, where $\sigma_j = \gamma_{\left([j],\frac{j}{[j]}\right)}$.
  \item $U[\vec{A}_{j}/\vec{X}] \equiv V_{\left([j],\frac{j}{[j]}\right)}$.
  \end{itemize}
  Take $f(a) = \sug{e}{a-1} d_e$, so we rewrite $S_a \equiv \sug{c}{d_{a}}\gamma_{(a,c)}\cdot V_{(a,c)}$ as
  \[
    S_a \equiv \sum^{f(a) + d_a}_{g = f(a)} \sigma_g\cdot V_{\left([g],\frac{g}{[g]}\right)}
    \equiv \sum^{f(a) + d_a}_{g = f(a)} \sigma_g\cdot U[\vec{A}_{g}/\vec{X}]
  \]
  Applying $\to_E$ for all $a \in \{1,\dots,b\}$,
  \[
    \prftree[r]{$\to_E$}
    {\Gamma\vdash \ve{t}: \sui{n}{\alpha_{i} \cdot\forall\vec{X}.(U\to T_{i})}}
    {\Gamma \vdash \ve{r}: \sum^{f(a) + d_a}_{g = f(a)} \sigma_g\cdot U[\vec{A}_{g}/\vec{X}]}
    {\Gamma \vdash (\ve{t})~\ve{r}: \sui{n}\sum^{f(a) + d_a}_{g = f(a)} \left(\alpha_i \times \sigma_g\right)\cdot T_{i}[\vec{A}_{g}/\vec{X}]}
  \]
  We rewrite $\sui{n}\sum^{f(a) + d_a}_{g = f(a)} \left(\alpha_i \times \sigma_g\right)\cdot T_{i}[\vec{A}_{g}/\vec{X}] \equiv P_a$,
  then by applying the $S$ rule we have
  \vspace{1cm}
  \[
    \prftree[r]{$S$}
    {\Gamma \vdash (\ve{t})~\ve{r}: P_a~\forall a \in \{1,\dots,b\}}
    {\Gamma \vdash \alpha\cdot(\ve{t})~\ve{r}: \sug{a}{b} \eta_a \cdot P_a}
  \]
  Now we begin to unravel the final result
  \begin{align*}
    \sug{a}{b} \eta_a \cdot P_a &\equiv \sug{a}{b} \eta_a \cdot \sui{n}\sum^{f(a) + d_a}_{g = f(a)} \left(\alpha_i \times \sigma_g\right)\cdot T_{i}[\vec{A}_{g}/\vec{X}]\\
                                &\equiv \sug{a}{b}\sum^{f(a) + d_a}_{g = f(a)}\suj{m} \left(\alpha_i \times \eta_{[g]} \times \sigma_g\right)\cdot T_{i}[\vec{A}_{g}/\vec{X}]\\
                                &\equiv \sug{a}{b}\sum^{f(a) + d_a}_{g = f(a)}\suj{m} (\alpha_i\times\beta_{g})\cdot T_{i}[\vec{A}_{g}/\vec{X}]\\
                                &\equiv \sui{n}\suj{m} (\alpha_{i}\times\beta_{j})\cdot T_{i}[\vec{A}_{j}/\vec{X}]\\
  \end{align*}
  Then,
  \[
    \Gamma \vdash \alpha\cdot(\ve{t})~\ve{r}: \sui{n}\suj{m} (\alpha_{i}\times\beta_{j})\cdot T_{i}[\vec{A}_{j}/\vec{X}]
  \]
  Since $\sui{n_k}\suj{m_k} (\alpha_{(k,i)}\times\beta_{(k,j)})\cdot T_{(k,i)}[\vec{A}_{(k,j)}/\vec{X}] \ssubt_{\V_k,\Gamma} R_k$,
  then for all $k \in \{1,\dots,h\}$, $\Gamma \vdash \alpha\cdot(\ve{t})~\ve{r}: R_k$.\\
  By applying the $S$ and $1_E$ rules, then 
  \[
    \prftree[r]{$1_E$}
    {\prftree[r]{$S$}
      {\Gamma\vdash\alpha \cdot(\ve t)~\ve r: R_k~\forall k \in \{1,\dots,h\}}
      {\Gamma\vdash 1\cdot(\alpha \cdot(\ve t)~\ve r): \suk{h} \mu_k \cdot R_k}}
    {\Gamma\vdash \alpha \cdot(\ve t)~\ve r: \suk{h} \mu_k \cdot R_k}
  \]
  Finally, by the $\equiv$ rule, then $\Gamma\vdash \alpha \cdot (\ve t)~\ve r:
  T$.
\end{proof}
}

\arxiv{\section{Omitted proofs in Section~\ref{ch:other-properties}}\label{app:proofsOP}}
\conf{\section{Proof of Theorem~\ref{thm:sn}}\label{app:proofsOP}}
\arxiv{\xrecap{Theorem}{Progress}{thm:progress}{
  Given $\mathbb{V} = \left\{\sui{n} \alpha_i \cdot \lambda x_i.\ve{t}_i + \sum^{m}_{j=n+1} \lambda x_j.\ve{t}_j \mid \forall i, j, \lambda x_i.\ve{t}_i \neq \lambda x_j.\ve{t}_j\right\}$
  and $\mathsf{NF}$ the set of terms in normal form (the terms that cannot be reduced any further), then
  if $\vdash \ve{t}: T$ and $\ve{t} \in \mathsf{NF}$, it follows that $\ve{t} \in \mathbb{V}$.
}
\begin{proof}
  By induction on $\ve{t}$:
  \inductioncase{Case $\ve{t} = \sui{n} \alpha_i \cdot \lambda x_i.\ve{t}_i + \sum^{m}_{j=n+1} \lambda x_j.\ve{t}_j \mid \forall i, j, \lambda x_i.\ve{t}_i \neq \lambda x_j.\ve{t}_j$}
  Trivial case.
  \inductioncase{Case $\ve{t} = \sui{n} \alpha_i \cdot \lambda x_i.\ve{t}_i + \sum^{m}_{j=n+1} \lambda x_j.\ve{t}_j \mid \exists i, j, \lambda x_i.\ve{t}_i = \lambda x_j.\ve{t}_j$}
  $\ve{t} \notin \mathsf{NF}$, since at least one reduction rule from Group F can be applied.
  \inductioncase{Case $\ve{t} = (\ve{r})~\ve{s}$}
  By induction hypothesis, we know that $\ve{r} = \sui{n} \alpha_i \cdot \lambda x_i.\ve{t}_i + \sum^{m}_{j=n+1} \lambda x_j.\ve{t}_j \in \mathbb{V}$.
  We consider the following cases:
  \begin{itemize}
  \item If $m>n+1$ or $n \neq 0$, then at least one reduction rule from Group A can be applied, hence $(\ve{r})~\ve{s} \notin \mathsf{NF}$.
  \item If $m=n+1$ and $n=0$, then $\ve{r} = \ve{b}_{n+1} \in \mathbb{V}$.
    Since $FV(\ve{r}) = \emptyset$, then $\ve{r} = \lambda x.\ve{r'}$, which implies $(\ve{r})~\ve{s}$ is a beta-redex
    or at least one reduction rule from Group A can be applied, hence $(\ve{r})~\ve{s} \notin \mathsf{NF}$.
  \end{itemize}
  \inductioncase{Case $\ve{t} = \alpha \cdot \ve{r}$}
  By induction hypothesis, we know that $\ve{r} = \sui{n} \alpha_i \cdot \lambda x_i.\ve{t}_i + \sum^{m}_{j=n+1} \lambda x_j.\ve{t}_j \in \mathbb{V}$.
  We consider the following cases:
  \begin{itemize}
  \item If $m \neq n+1$ or $n \neq 0$, then at least one reduction rule from Group E can be applied, hence $(\ve{r})~\ve{s} \notin \mathsf{NF}$.
  \item If $m=n+1$, $n=0$ and $\alpha = 1$, then $\ve{r} = \lambda x.\ve{t} \in \mathbb{V}$, but $1 \cdot \ve{r} = 1 \cdot \lambda x.\ve{t} \to \lambda x.\ve{t}$, hence $\alpha \cdot \ve{r} \notin \mathsf{NF}$.
  \item If $m=n+1$, $n=0$ and $\alpha \neq 1$, then $\ve{r} = \lambda x.\ve{t} \in \mathbb{V}$ and $\alpha \cdot \ve{r} = \alpha \cdot \ve{b} \in \mathbb{V}$.
  \end{itemize}
  \inductioncase{Case $\ve{t} = \ve{t}_1 + \ve{t}_2$}
  By induction hypothesis, we know that $\ve{t}_{k} = \sui{n^k} \alpha^k_i \cdot (\lambda x_i.\ve{t}_i)^k_i + \sum^{m^k}_{j=n+1} (\lambda x_j.\ve{t}_j)^k_j \in \mathbb{V}$, with $k = 1, 2$.\\
  We consider the following cases:
  \begin{itemize}
  \item $\exists i, j ~/~ (\lambda x_i.\ve{t}_i)^1 = (\lambda x_j.\ve{t}_j)^2$, then at least one reduction rule from Group F can be applied, hence $\ve{t}_1 + \ve{t}_2 \notin \mathsf{NF}$.
  \item $\forall i, j ~/~ (\lambda x_i.\ve{t}_i)^1 \neq (\lambda x_j.\ve{t}_j)^2$, then by definition of $\mathbb{V}$, $\ve{t}_1 + \ve{t}_2 \in \mathsf{NF}$.\qed
  \end{itemize}
\end{proof}
}

\xrecap{Theorem}{Strong Normalisation}{thm:sn}
{If $\Gamma \vdash \ve{t}: T$ is a valid sequent, then $\ve{t}$ is strongly normalising.}

\begin{proof}
  Consider the following derivation tree in \lvecr, where $T \equiv T_1$,
  \[
    \pi_1 = \left\{\vcenter{
        \prftree
        {\vdots}
        {\Gamma \vdash \ve{t}: T_1}}\right.
  \]
  Since the only difference between $\lvecr$ and $\lvec$ is the replacement of
  the $\alpha_I$ rule for the $S$ and $1_E$ rules, then if $S$ and $1_E$ are not
  present in $\pi_1$, we have that $\pi_1$ (and particularly, $\Gamma \vdash
  \ve{t}: T$) is also a valid derivation for $\lvec$. Also, notice that up to this
  point, neither the term nor the types have scalars associated with them. If a
  scalar were to be introduced, then the derivation trees (for $\lvecr$ and
  $\lvec$) would be
  \[
    \begin{array}{c@{\qquad}c}
      \text{In }\lvecr
      &
        \text{In }\lvec
      \\[2ex]
      \vcenter{\prftree[r]{$S$}
      {\pi_1}
      {\dots}
      {\pi_n}
      {\Gamma \vdash \left(\sui{n} \alpha_i\right) \cdot \ve{t}: \sui{n} \alpha_i \cdot T_i}}
      &
        \vcenter{\prftree[r]{$\alpha_I$}
        {\pi_1}
        {\Gamma \vdash \left(\sui{n} \alpha_i\right) \cdot \ve{t}: \left(\sui{n} \alpha_i\right) \cdot T}}
    \end{array}
  \]
  Where $\pi_i = \Gamma \vdash \ve{t}: T_i$ are valid sequents for some $T_i$,
  with $i \in \{2, \dots, n\}$.\\ Now, notice that by having $\Gamma \vdash
  \left(\sui{n} \alpha_i\right) \cdot \ve{t}: \sui{n} \alpha_i \cdot T_i$
  (specifically, by having a linear combination of types), we are restricting the
  terms we can type. In other words, for every derivation tree in $\lvecr$, there
  is a simpler derivation tree in $\lvec$, and thus if a sequent $\Gamma \vdash
  \ve{t}: T$ is valid in $\lvecr$, then there is a derivation tree for the same
  term in $\lvec$. Finally, since $\lvec$ is strongly
  normalising~\cite[Thm.~5.7]{vectorial}, then $\lvecr$ is strongly normalising.
\end{proof}

\arxiv{
\recap{Lemma}{lem:wp:weightequiv}{
  If $T \equiv R$, then $\tnorm{T} = \tnorm{R}$.
}
\begin{proof}
  We prove the lemma holds for every definition of $\equiv$
  \inductioncase{Case $1\cdot T \equiv T$}
  Trivial case.
  \inductioncase{Case $\alpha\cdot(\beta\cdot T) \equiv (\alpha\times\beta)\cdot T$}
  \vspace{-0.7cm}
  \begin{align*}
    \tnorm{\alpha\cdot(\beta\cdot T)} &= \alpha\cdot\tnorm{\beta\cdot T} = (\alpha \times \beta) \cdot \tnorm{T}
                                        = \tnorm{(\alpha \times \beta) \cdot T}
  \end{align*}
  \inductioncase{Case $\alpha\cdot T+\alpha\cdot R \equiv \alpha\cdot (T+R)$}
  \vspace{-0.7cm}
  \begin{align*}
    \tnorm{\alpha\cdot T+\alpha\cdot R} &= \tnorm{\alpha\cdot T} + \tnorm{\alpha\cdot R}\\
                                        &= \alpha \cdot \tnorm{T} + \alpha \cdot \tnorm{R} = \alpha \cdot (\tnorm{T} + \tnorm{R})\\
                                        &= \alpha \cdot (\tnorm{T + R}) = \tnorm{\alpha\cdot (T+R)}\\
  \end{align*}
  \inductioncase{Case $\alpha\cdot T+\beta\cdot T \equiv (\alpha+\beta)\cdot T$}
  \vspace{-0.7cm}
  \begin{align*}
    \tnorm{\alpha\cdot T+\beta\cdot T} &= \tnorm{\alpha \cdot T} + \tnorm{\beta \cdot T} = \alpha \cdot \tnorm{T} + \beta \cdot \tnorm{T}\\
                                       &= (\alpha + \beta) \cdot \tnorm{T} = \tnorm{(\alpha+\beta)\cdot T}
  \end{align*}
  \inductioncase{Case $T+R \equiv R+T$}
  \vspace{-0.7cm}
  \begin{align*}
    \tnorm{T + R} = \tnorm{T} + \tnorm{R} = \tnorm{R} + \tnorm{T} = \tnorm{T + R}
  \end{align*}
  \inductioncase{Case $T+(R+S) \equiv (T+R)+S$}
  \vspace{-0.7cm}
  \begin{align*}
    \tnorm{T + (R+S)} &= \tnorm{T} + \tnorm{R + S} = \tnorm{T} + \tnorm{R} + \tnorm{S}\\
                      &= \tnorm{T + R} + \tnorm{S} = \tnorm{(T + R) + S} \qed
  \end{align*}
\end{proof}

\recap{Lemma}{lem:wp:weightofvalues}{
  If $\ve{v}\in \mathbb{V}$, and $\vdash \ve{v}: T$, then
  $\tnorm{T} \equiv \tnorm{\ve{v}}$.
}
\begin{proof}
  Let $\ve{v} = \sui{k} \alpha_i \cdot \lambda x_i.\ve{t}_i + \sum^{n}_{i = k+1} \lambda x_i.\ve{t}$.
  We proceed by induction on $n$.
  \inductioncase{Case $n = 1$}
  There are two possible escenarios:
  \textleadbydots{\textbf{$k = 1$}}
  \noindent In this scenario, consider $\pi =~\vdash \alpha_1 \cdot \lambda x_1.\ve{t}_1: T$.
  By Lemma~\ref{lem:sr:scalars}, there exist $R_1, \dots, R_m$, $\beta_1, \dots, \beta_m$ such that
  \begin{itemize}
  \item $T \equiv \suj{m}\beta_j \cdot R_j$.
  \item $\pi_i = ~ \vdash \lambda x_1.\ve{t}_1: R_j$, with $size(\pi) > size(\pi_j)$, for $j \in \{1, \dots, m\}$.
  \item $\sui{m} \beta_j = \alpha_1$.
  \end{itemize}
  Considering $\lambda x_1.\ve{t}_1$ is a basis term, then by Lemma~\ref{lem:sr:basevectors},
  for each $j \in \{1, \dots, m\}$ (we will omit the $j$ index for readability),
  there exist $U_1, \dots, U_h$, $\sigma_1, \dots, \sigma_h$ such that
  \begin{itemize}
  \item $R \equiv \suk{h} \sigma_k \cdot U_k$.
  \item $\vdash \lambda x_1.\ve{t}_1: U_k$, for $k \in \{1,\dots,h\}$.
  \item $\suk{h} \sigma_k = 1$.
  \end{itemize}
  Then,
  \[
    T \equiv \suj{m} \beta_j \cdot R_j \equiv \suj{m} \beta_j \cdot (\suk{h_j} \sigma_{(j,k)} \cdot U_{(j,k)})
  \]
  Finally, by definition of $\tnorm{\bullet}$, we have
  \begin{align*}
    \tnorm{\ve{v}} &= \tnorm{\alpha_1 \cdot \lambda x_1.\ve{t}_1} = \alpha_1 \cdot \tnorm{\lambda x_1.\ve{t}_1}\\
                   &= \alpha_1 = \sui{m} \beta_j = \sui{m} \beta_j \cdot \underbrace{\left(\suk{h_j} \sigma_{(j,k)}\right)}_{=~1}\\
                   &= \sui{m} \beta_j \cdot \left(\suk{h_j} \sigma_{(j,k)} \cdot \tnorm{U_{(j,k)}}\right)\\
                   &= \sui{m} \beta_j \cdot \tnorm{\suk{h_j} \sigma_{(j,k)} \cdot U_{(j,k)}}
                     = \tnorm{\sui{m} \beta_j \cdot \left(\suk{h_j} \sigma_{(j,k)} \cdot U_{(j,k)}\right)}\\
                   &= \tnorm{T}
  \end{align*}
  \textleadbydots{\textbf{$k = 0$}}
  \noindent In this scenario, consider $\vdash \lambda x_1.\ve{t}_1: T$.
  Considering $\lambda x_1.\ve{t}_1$ is a basis term, then
  by Lemma~\ref{lem:sr:basevectors}, there exist $U_1, \dots, U_m$, $\beta_1, \dots, \beta_m$ such that
  \begin{itemize}
  \item $T \equiv \suj{m} \beta_k \cdot U_j$.
  \item $\vdash \lambda x_1.\ve{t}_1: U_j$, for $j \in \{1,\dots,m\}$.
  \item $\suj{m} \beta_j = 1$.
  \end{itemize}
  Finally, by definition of $\tnorm{\bullet}$, we have
  \begin{align*}
    \tnorm{\ve{v}} &= \tnorm{\lambda x_1.\ve{t}_1} = 1 = \suj{m} \beta_j\\
                   &= \suj{m} \beta_j \cdot \tnorm{U_j} = \tnorm{\suj{m} \beta_j \cdot U_j}\\
                   &= \tnorm{T}
  \end{align*}
  \inductioncase{Induction step}
  Consider now that $\vdash \ve{v} = \ve{v}' + \ve{v}'': T$,
  where $\ve{v}' = \sui{k} \alpha_i \cdot \lambda x_i.\ve{t}_i + \sum^{n}_{i = k+1} \lambda x_i.\ve{t}_i$
  and either $\ve{v}'' = \beta \cdot \lambda x.\ve{t}$, or $\ve{v}'' = \lambda x.\ve{t}$.
  By Lemma~\ref{lem:sr:sums}, we know there exists $R$ and $S$ such that
  \begin{itemize}
  \item $T \equiv R + S$.
  \item $\Gamma\vdash\ve{v}': R$.
  \item $\Gamma\vdash\ve{v}'': S$.
  \end{itemize}
  By induction hypothesis, since $\vdash \ve{v}' = \sui{k} \alpha_i \cdot \lambda x_i.\ve{t}_i + \sum^{n}_{i = k+1} \lambda x_i.\ve{t}_i: R$,
  then $\tnorm{R} = \tnorm{\ve{v}'}$;
  and since either $\ve{v}'' = \beta \cdot \lambda x.\ve{t}$ or $\ve{v}'' = \lambda x.\ve{t}$, in both cases we know that $\tnorm{S} = \tnorm{\ve{v}''}$.
  Finally, and considering by Lemma~\ref{lem:wp:weightequiv} that $\tnorm{T} = \tnorm{R} + \tnorm{S}$, we have
  \begin{align*}
    \tnorm{\ve{v}} &= \tnorm{\ve{v}' + \ve{v}''}\\
                   &= \tnorm{\ve{v}'} + \tnorm{\ve{v}''}\\
                   &= \tnorm{R} + \tnorm{S}\\
                   &= \tnorm{T}
                     \qed
  \end{align*}
\end{proof}
}

\end{document}